\documentclass[journal,twoside,web]{ieeecolor}
\usepackage{generic}
\usepackage{cite}
\usepackage{amsmath,amssymb,amsfonts}
\usepackage{graphicx}
\usepackage{hyperref}
\usepackage{textcomp}

\usepackage[ruled,vlined]{algorithm2e}
\usepackage{bm}
\usepackage{booktabs}
\usepackage{color}
\usepackage{multirow}
\usepackage{makecell}
\usepackage{subfigure}
\usepackage{cases}

\newtheorem{theorem}{Theorem}
\newtheorem{lemma}{Lemma}
\newtheorem{definition}{Definition}
\newtheorem{remark}{Remark}
\newtheorem{proposition}{Proposition}
\newtheorem{assumption}{Assumption}
\newtheorem{corollary}{Corollary}

\newtheorem{example}{Example}

\newtheorem{problem}{Problem}

\usepackage{array}
\renewcommand\arraystretch{0.75}

\def\BibTeX{{\rm B\kern-.05em{\sc i\kern-.025em b}\kern-.08em
    T\kern-.1667em\lower.7ex\hbox{E}\kern-.125emX}}
\begin{document}
	
%\title{Resilient Leader-Follower Consensus in Time-Varying Networks with Multi-Hop Relays}

%\title{Leader-Follower Consensus in Time-Varying Networks: A Resilient Multi-Hop Relay Approach}

\title{Reaching Resilient Leader-Follower Consensus in Time-Varying Networks via Multi-Hop Relays}

\author{Liwei Yuan, \IEEEmembership{Member, IEEE}, and Hideaki Ishii, \IEEEmembership{Fellow, IEEE}
	\thanks{This work was supported in part by JSPS under Grant-in-Aid for 
		Scientific Research Grant No.~22H01508.}
	\thanks{L. Yuan is with the College of Electrical and Information Engineering, Hunan University, Changsha, 410082, China (e-mail: yuanliwei@hnu.edu.cn). }
	\thanks{H. Ishii is with the Department of Information Physics and Computing, The University of Tokyo, Tokyo, 113-8656, Japan (e-mail: hideaki\_ishii@ipc.i.u-tokyo.ac.jp). }
}

\maketitle

\begin{abstract}
We study the problem of resilient leader-follower consensus of multi-agent systems (MASs) in the presence of adversarial agents, where agents' communication is modeled by time-varying topologies. The objective is to develop distributed algorithms for the nonfaulty/normal follower agents to track an arbitrary reference value propagated by a set of leaders while they are in interaction with the unknown adversarial agents.
Our approaches are based on the weighted mean subsequence reduced (W-MSR) algorithms with agents being capable to communicate with multi-hop neighbors.
The proposed algorithms solve our resilient leader-follower consensus problem with agents possessing first-order and second-order dynamics.
Moreover, we characterize tight necessary and sufficient graph conditions for the proposed algorithms to succeed in terms of the novel notion of jointly robust following graphs. Our graph condition is tighter than the sufficient graph conditions in the literature when agents use only one-hop communication (without relays). Using multi-hop relays, we are able to enhance robustness of leader-follower networks without increasing physical communication links and obtain further relaxed graph requirements for our algorithms to succeed. Numerical examples are given to verify the efficacy of the proposed algorithms.
\end{abstract}

\begin{IEEEkeywords}
Cyber security, leader-follower network, resilient algorithms, time-varying topology.
\end{IEEEkeywords}

\section{Introduction}

\IEEEPARstart{O}{ver}
the past few decades, distributed consensus has emerged as a cornerstone of research in the fields of multi-agent systems (MASs) and distributed algorithms \cite{bullo2009distributed,Lynch,ren2005consensus}. In such a problem, agents connected over a network try to reach consensus on a common value while interacting with only neighboring agents. Stemming from this concept, extensive applications and algorithms have been devised to overcome various industrial challenges \cite{iori2024resilient,yang2013consensus,mitra2021,nedic2009distributed}.
Concurrently, growing concerns over cyber security within MASs have amplified the significance of consensus protocols, especially in scenarios where adversarial agents induce failures or launch attacks, e.g.,
\cite{sundaram2018distributed,ishii2022overview,sundaram2011distributed,yuan2021secure}.
Under this topic, the problems of resilient consensus have drawn much attention in areas of systems control, distributed computing, and cooperative robotics \cite{dibaji2017resilient, leblanc2013resilient, vaidya2012iterative,su2017reaching,yu2022resilient}, where the nonfaulty, normal agents aim to reach consensus despite the possible misbehaviors of adversarial agents. A common goal in this setting is that normal agents arrive at the same value located within the convex hull of their initial states. However, for applications such as formation control and reliable broadcast, it is desirable that agents together track a specific value which is externally given and may be outside such a convex hull.
Thus, it motivates us to extend resilient consensus algorithms for such objectives.

A related problem in prior literature is the leader-follower consensus problem, where the goal is for normally behaving agents to come to an agreement on the reference value of a leader or a set of leaders \cite{ren2008consensus,dimarogonas2009leader}. However, these works considered MASs without any adversarial agents, potentially rendering them vulnerable to random failures or deliberate attacks. Within the domain of distributed computing, considerable efforts have been dedicated to ensuring reliable broadcast \cite{koo2006reliable} as well as the certified propagation algorithm (CPA) \cite{koo2004broadcast,tseng2015broadcast}. In these works, the objective is for a secure leader to broadcast a reference value to all nodes in the network in the presence of adversarial agents. Additionally, there is a body of research addressing the problem known as resilient distributed estimation. For instance, the work \cite{leBlanc2014resilient} studied resilient parameter estimation, where certain reliable agents drive the errors of the remaining normal agents to the static reference value of zero. Moreover, the authors of \cite{mitra2019byzantine} investigated a problem where the observation information of the system is resiliently transmitted from a group of source nodes to other nodes that cannot directly observe the system.

In this paper, we develop distributed algorithms to tackle resilient leader-follower consensus in time-varying networks.
In the literature, many efforts have been devoted to resilient consensus using the so-called mean subsequence reduced (MSR) algorithms \cite{vaidya2012iterative,leblanc2013resilient,usevitch2020determining,senejohnny2019resilience,yuan2023event}.
In such algorithms, each normal agent disregards the most deviated states of neighbors to avoid being affected by possible faulty values from adversarial neighbors.
Tight graph conditions on static (i.e., time-invariant) network structures guaranteeing the success of MSR algorithms have been derived for the class of malicious agents \cite{leblanc2013resilient,abbas2017improving,dibaji2017resilient} as well as the class of Byzantine agents \cite{vaidya2012iterative,yuan2022asynchronous}. Notably, \cite{leblanc2013resilient} demonstrated that static networks utilizing MSR algorithms must adhere to a specific structural criterion, called graph robustness, to attain resilient consensus. However, the majority of these studies have been confined to static networks, i.e., communication topologies are fixed across iterations.
However, in numerous applications of MASs with physical motions, e.g., formation control of drones and vehicle platoons, the underlying communication network may be time-varying due to limited communication ranges and temporal variations of communication channels \cite{bullo2009distributed,wen2023joint,yuan2024timevarying}.
Thus, there is a significant demand for investigating resilient leader-follower consensus in time-varying networks. For instance, the work \cite{usevitch2020resilient} proved a sufficient condition for the sliding weighted-MSR (SW-MSR) algorithm from \cite{saldana2017resilient} to achieve resilient leader-follower consensus to arbitrary static reference values. It reduced the stringent connectivity requirements of MSR algorithms at each iteration. Later, \cite{rezaee2021resiliency} studied resilient leader-follower consensus in static networks with the leader in each network having a dynamic reference value.

Meanwhile, several works relaxed the graph connectivity requirements for resilient consensus in static networks through multi-hop communication \cite{su2017reaching, sakavalas2020asynchronous,yuan2022asynchronous}, which enables messages sent by an agent to reach beyond its direct neighbors through relays by middle agents \cite{Lynch,goldsmith2005wireless}. It can improve network resilience against adversaries without changing the original topology as shown in \cite{su2017reaching, sakavalas2020asynchronous,yuan2021resilient}.
Motivated by these works, we are interested to investigate whether multi-hop relays could further help us to acquire a more relaxed condition for leader-follower consensus in time-varying networks.

We summarize the contributions of this paper as follows. 
First, we characterize a necessary and sufficient graph condition for the Multi-hop Weighted-MSR (MW-MSR) algorithm to achieve resilient leader-follower consensus in time-varying networks. Consequently, the normal follower agents are able to track the reference value propagated by a set of leaders in the presence of Byzantine agents, which may also include adversarial leaders.
Our graph condition is denoted by a novel notion of jointly robust following graphs with multi-hop communication. Compared to the SW-MSR algorithm \cite{saldana2017resilient, usevitch2020resilient} storing neighbors' values for the last certain time steps at each iteration, our approach utilizes neighbors' values of only the current time step at each iteration.
It is notable that even with one-hop communication, our graph condition is tighter than the ones in the resilient leader-follower consensus works with static reference values \cite{usevitch2020resilient} as well as dynamic reference values \cite{rezaee2021resiliency}.
Moreover, by increasing the number of relaying hops, our method can increase the graph robustness against adversaries without changing the network topology. Hence, our approach can tolerate more adversarial nodes compared to the one-hop MSR algorithms \cite{leblanc2013resilient,rezaee2021resiliency,ishii2022overview,wen2023joint} as well as the CPA-based methods \cite{koo2004broadcast,tseng2015broadcast}.
Moreover, numerical examples show that our method can achieve resilient leader-follower consensus in sparse time-varying networks where the algorithms in \cite{usevitch2020resilient, rezaee2021resiliency} have difficulties. As a side result, we present that the tight graph condition for resilient leader-follower consensus under the malicious model is the same as the one for the Byzantine model, even though malicious agents are less adversarial.

Second, we also deal with resilient leader-follower consensus in time-varying networks for agents with second-order dynamics and propose a multi-hop double-integrator position-based MSR (MDP-MSR) algorithm. 
This extension is vital since double-integrator dynamics are often used to characterize more accurate motions of agents in robotics; see, e.g., \cite{chipade2021multiagent}.
To the best of our knowledge, such a problem has not been investigated in the literature.
Furthermore, we derive a necessary and sufficient graph condition for the MDP-MSR algorithm to handle this case. The condition is the same as the one for the MW-MSR algorithm. Moreover, we provide necessary properties for verifying whether network topologies meet our conditions or not. 
Both theoretical results and numerical examples verify that the proposed algorithm with multi-hop relays can improve the robustness against adversaries in static as well as time-varying networks for agents with second-order dynamics.
%Compared to the dynamic leader-follower consensus works for the fault-free case \cite{ren2005consensus} and the resilient case \cite{rezaee2021resiliency}, our algorithm 
Lastly, we apply the algorithm for achieving formation control in the leader-follower configuration in the presence of adversaries, which could serve as a basis for applications of, e.g., multi-robot manufacturing in complex industrial sectors.

The rest of this paper is organized as follows. 
In Section~II, we outline the problem settings. In Section~III, we define the novel notion of joint robust following graphs with multi-hop communication.
In Section~IV, we derive a tight graph condition under which the MW-MSR algorithm guarantees resilient leader-follower consensus.
In Section~V, we introduce the MDP-MSR algorithm for MASs with second-order dynamics and provide tight graph conditions for the algorithm to achieve resilient leader-follower consensus in static and time-varying networks.
In Section~VI, we present numerical examples to verify the efficacy of our algorithms in sparse time-varying networks.
Finally, we conclude the paper in Section~VII. 
Compared to the preliminary version of this work \cite{yuan2024resilient}, the current paper contains additional results for time-varying networks, the results for the secure leader, the results for second-order MASs, and extensive numerical examples.

\section{Preliminaries and Problem Settings}

\subsection{Graph Notions}
Consider a time-varying directed graph $\mathcal{G}[k] = (\mathcal{V},\mathcal{E}[k])$ consisting of the node set $\mathcal{V}=\{1,...,n\}$ and the time-varying edge set $\mathcal{E}[k]$. The edge $(j,i)\in \mathcal{E}[k]$ indicates that node $i$ can get information from node $j$ at time $k\in \mathbb{Z}_{\geq 0} $.
The union of the graphs $\mathcal{G}[k] = (\mathcal{V},\mathcal{E}[k])$ across the time interval $[k_1,k_t]$ is denoted by $\overline{\mathcal{G}} = (\mathcal{V},\overline{\mathcal{E}})$, where $\overline{\mathcal{E}}=\bigcup_{j=1}^{t}\mathcal{E}[k_j]$.
The subgraph of $\mathcal{G}[k] = (\mathcal{V},\mathcal{E}[k])$ induced by the node set $\mathcal{H}\subset\mathcal{V}$ is the subgraph $\mathcal{G}_\mathcal{H}[k]=(\mathcal{V}(\mathcal{H}),\mathcal{E}(\mathcal{H})[k])$, where $\mathcal{V}(\mathcal{H})=\mathcal{H}$ and $\mathcal{E}(\mathcal{H})[k]=\{(i,j)\in \mathcal{E}[k]: i,j\in \mathcal{H}\}$.

An $l$-hop path from source node $i_1$ to destination node $i_{l+1}$ is a sequence of distinct nodes $(i_1, i_2, \dots, i_{l+1})$, where $(i_j, i_{j+1})\in \mathcal{E}[k] $ for $j=1, \dots, l$. Node $i_{l+1}$ is said to be reachable from node $i_1$ at time $k$. 
Let $\mathcal{N}_i^{l-}[k]$ be the set of nodes that can reach node $i$ via paths of at most $l$ hops at time $k$.
Let $\mathcal{N}_i^{l+}[k]$ be the set of nodes that are reachable from node $i$ via paths of at most $l$ hops at time $k$. Node $i$ is included in both sets above.
The $l$-th power of the graph $\mathcal{G}[k]$, denoted by $\mathcal{G}^l[k]$, is a multigraph with $\mathcal{V}$ and a directed edge from node $j$ to node $i$ is defined by a path of length at most $l$ from $j$ to $i$ in $\mathcal{G}[k]$. The adjacency matrix $A[k] = [a_{ij}[k] ]$ of $\mathcal{G}^l[k]$ is given by $\alpha \leq a_{ij}[k]<1$ if $j\in \mathcal{N}_i^{l-}[k]$ and otherwise $a_{ij}[k] = 0$, where $\alpha > 0$ is fixed and $\sum_{j=1,j\neq i}^{n} a_{ij}[k]\leq 1, \forall k$. 
%Let $L[k] = [b_{ij}[k] ]$ be the Laplacian matrix of $\mathcal{G}^l[k]$, where $b_{ii}[k] =\sum_{j=1,j\neq i}^{n}a_{ij}[k]$, $b_{ij}[k] = -a_{ij}[k]$ for $ i\neq j$.

Next, we describe our communication model. Node $i_1$ can send its own messages to an $l$-hop neighbor $i_{l+1}$ via different paths at time $k$.
We represent a message as a tuple $m=(w,P)$, where $w=\mathrm{value}(m)\in \mathbb{R}$ is the message content and $P=\mathrm{path}(m)$ indicates the path via which $m$ is transmitted. 
At time $k\geq 0$, each normal node $i$ exchanges the messages $m_{ij}[k]=(x_i[k],P_{ij}[k])$ consisting of its state $x_i[k]\in \mathbb{R}$ along each path $P_{ij}[k]$ with its multi-hop neighbor $j$ via the relaying process in \cite{yuan2021resilient}.
Denote by $\mathcal{V}(P)$ the set of nodes in $P$.

\subsection{System Model and Algorithm}\label{problemsetting}
In our leader-follower consensus problem, we consider the time-varying MAS modeled by the graph $\mathcal{G}[k] = (\mathcal{V},\mathcal{E}[k])$, where $\mathcal{V}$ consists of the set of leader agents $\mathcal{L}$ and the set of follower agents $\mathcal{W}$ with $\mathcal{L}\cup \mathcal{W}=\mathcal{V}$ and $\mathcal{L}\cap \mathcal{W}=\emptyset$. Leader agents in $\mathcal{L}$ propagate a desired reference scalar value to follower agents in $\mathcal{W}$, and thereafter, follower agents achieve consensus on that reference value. However, during the propagation, if adversarial agents are in presence, they may misbehave and try to prevent normal agents from reaching leader-follower consensus. 

To characterize our system under attacks, we denote the set of adversarial agents by $\mathcal{A}\subset \mathcal{V}$ and denote the set of non-adversarial, normal agents by $\mathcal{N}=\mathcal{V}\setminus\mathcal{A}$ with $n_N=|\mathcal{N}|$. Formal definitions of adversarial agents are given later. Then, the sets of normal leader agents and normal follower agents are denoted by $\mathcal{L}^\mathcal{N}=\mathcal{L}\cap \mathcal{N}$ and $\mathcal{W}^\mathcal{N}=\mathcal{W}\cap \mathcal{N}$, respectively.

\begin{table}[t]\label{table1}
	\caption{Types of leader agents.}
	\renewcommand\arraystretch{1.5}
	\centering
	\begin{tabular}{c|c|c} 
		\hline
		             & Known by followers & Not known by followers \\  
		\hline
		Secure      &  Case 1: Proposition~\ref{proposition_secure}           & Case 2: Theorem~\ref{theorem_firstorder} \\
		\hline
		Not secure  &  Case 3: Theorem~\ref{theorem_firstorder}           & Case 4: Theorem~\ref{theorem_firstorder} \\
		\hline
	\end{tabular}
	\vspace*{-4mm}
\end{table}

Before we present our system model, we introduce an important categorization of the types of leader agents in the literature, e.g., \cite{usevitch2020resilient,rezaee2021resiliency}. For the leader agents, there are four cases depending on if they are secure (i.e., no faults) or not, and if they are known by followers or not (see Table~I). In this paper, we mainly focus on the cases where leader agents are not secure (i.e., Cases 3 and 4). Then we will give an analysis for the cases where leader agents are secure (i.e., Cases 1 and 2). Such an analysis is closely related to the one for the insecure leader cases.

%\begin{remark}
%	4 cases: 
%	If leaders are guaranteed to be normal or not
%	If leaders are known to followers or not
%	
%	Cases 1 and 2: If leaders are not guaranteed to be normal (known and not known), we use the MW-MSR algorithm and the corresponding condition
%	
%	Case 3: If leaders are normal and not known, we use the MW-MSR algorithm and there needs f+1 leaders and f+1 second layer nodes. (Each follower must have in-degree no less than f+1)
%	
%	Case 4: If leaders are normal and known, followers other than the second layer nodes use the MW-MSR algorithm, the condition follows the dynamic work
%\end{remark}

At each time $k$, each normal leader agent $d\in \mathcal{L}^\mathcal{N}$ updates its value according to a reference function $r[k] \in \mathbb{R}$ as
\begin{equation}\label{leader}
x_d[k + 1] = r [k],
\end{equation}
where $r[k]$ is assumed to be constant and the same for all normal leaders. We note that it can also be a staircase function and asymptotic tracking can be achieved. See Section~\ref{sec_example1}.

We define the resilient leader-follower consensus problem of this paper, which is also studied in \cite{usevitch2020resilient}.
\begin{problem}\label{problem}
	We say that the normal agents in $\mathcal{N}$ reach resilient leader-follower consensus if for any possible sets and behaviors of the adversaries in $\mathcal{A}$ and any state values of the normal agents in $\mathcal{N}$, the following condition is satisfied:
	\begin{equation}\label{reach_consensus}
		\lim_{k\to \infty}\medspace \max_{i\in \mathcal{W}^\mathcal{N}, \medspace d\in \mathcal{L}^\mathcal{N}} \medspace  |x_i[k]-x_d[k]|=0.
	\end{equation}
\end{problem}
\vspace{2mm}

To avoid being affected by adversarial agents, each normal follower agent $i$ updates its value according to the MW-MSR algorithm from \cite{yuan2021resilient}, which is presented in Algorithm~1. The notion of minimum message cover (MMC) \cite{su2017reaching} is crucial in Algorithm~1, which is defined as follows.

\begin{definition} For a graph $\mathcal{G} = (\mathcal{V},\mathcal{E})$, let $\mathcal{M}$ be a set of messages transmitted through $\mathcal{G}$, and let $\mathcal{P}(\mathcal{M})$ be the set of message paths of all the messages in $\mathcal{M}$, i.e., $\mathcal{P}(\mathcal{M}) =\{\mathrm{path}(m):m \in \mathcal{M}\}$. A \textit{message cover} of $\mathcal{M}$ is a set of nodes $\mathcal{T}(\mathcal{M})\subset \mathcal{V}$ whose removal disconnects all message paths, i.e., for each path $P\in \mathcal{P}(\mathcal{M})$, we have $\mathcal{V}(P)\cap \mathcal{T}(\mathcal{M})\neq \emptyset$. In particular, a \textit{minimum message cover} of $\mathcal{M}$ is defined by
	\begin{equation*}
		\mathcal{T}^*(\mathcal{M})\in	\arg \min_{\substack{ \mathcal{T}(\mathcal{M}): \textup{ Cover of } \mathcal{M}}} 	\left|  \mathcal{T} (\mathcal{M})\right| . 
	\end{equation*}
\end{definition}

\vspace{2mm}

\begin{algorithm}[t]
	\caption{MW-MSR Algorithm }
	\LinesNumbered 
	\KwIn{Node $i$ knows $x_i[0]$, $\mathcal{N}_i^{l-}[k]$, $\mathcal{N}_i^{l+}[k]$. }
	
	\For{$k\geq0$}{
		
		\SetKwBlock{newbox}{1) Exchange messages:}{}
		\newbox{
			\SetAlgoVlined
			Send $m_{ij}[k]=(x_i[k],P_{ij}[k])$ to $\forall j\in \mathcal{N}_i^{l+}[k]$. 
			
			Receive $m_{ji}[k]=(x_j[k],P_{ji}[k])$ from $\forall j\in \mathcal{N}_i^{l-}[k]$ and store them in $\mathcal{M}_i[k]$.
			
			Sort $\mathcal{M}_i[k]$ in an increasing order based on the message values (i.e., $x_j[k]$ in $m_{ji}[k]$).
		}
		
		\SetKwBlock{newbox}{2) Remove extreme values:}{}
		\newbox{
			\SetAlgoVlined
			
			(a) Define two subsets of $\mathcal{M}_i[k]$:
			\begin{equation*}
				\overline{\mathcal{M}}_i[k]=\{ m\in \mathcal{M}_i[k]: \mathrm{value}(m)> x_i[k]  \},
			\end{equation*}
			\begin{equation*}
				\underline{\mathcal{M}}_i[k]=\{ m\in \mathcal{M}_i[k]: \mathrm{value}(m)< x_i[k]  \}.
			\end{equation*}
			
			(b) Get $\overline{\mathcal{R}}_i[k]$ from $\overline{\mathcal{M}}_i[k]$:
			
			\vspace{1mm}
			
			\eIf{   $\left|  \mathcal{T}^* (\overline{\mathcal{M}}_i[k])\right| <f$ }{
				
				\vspace{1mm}
				$\overline{\mathcal{R}}_i[k]=\overline{\mathcal{M}}_i[k]$;
			}
			{
				Choose $\overline{\mathcal{R}}_i[k]$ s.t. (i)
				$\forall m\in \overline{\mathcal{M}}_i[k]\setminus \overline{\mathcal{R}}_i[k]$, $\forall m'\in \overline{\mathcal{R}}_i[k]$, $\mathrm{value}(m) \leq \mathrm{value}(m')  \medspace\medspace \medspace \medspace  $ 
				and (ii) $\left|  \mathcal{T}^* (\overline{\mathcal{R}}_i[k])\right| =f$. 
				
			}
			
			(c) Similar to (b), get $\underline{\mathcal{R}}_i[k]$ from $\underline{\mathcal{M}}_i[k]$, which contains smallest message values.
			
			(d) $\mathcal{R}_i[k]=\overline{\mathcal{R}}_i[k]\cup\underline{\mathcal{R}}_i[k]$.
		}
		
		\SetKwBlock{newbox}{3) Update:}{}
		\newbox{
			\SetAlgoVlined
			$\medspace\medspace \medspace\medspace\medspace a_{i}[k]=1/(\left| \mathcal{M}_i[k]\setminus \mathcal{R}_i[k] \right| )$,
			\begin{equation}
				x_i[k+1]=\sum_{m\in \mathcal{M}_i[k]\setminus \mathcal{R}_i[k]} a_{i}[k] \medspace\medspace \mathrm{value}(m).  \label{msrupdate}
			\end{equation}
		}
		\KwOut{$x_i[k+1]$.}
	}
	%\vspace*{-1.0mm}
\end{algorithm}
\vspace{-1.0mm}

In Algorithm~1, normal follower $i$ can remove the largest and smallest values from exactly $f$ nodes located within $l$ hops. With multi-hop relays, node $i$ might get multiple values from the same neighbor at each step. Thus, the MMC is needed for determining the number of extreme values to be removed. A more detailed explanation of Algorithm~1 can be found in \cite{yuan2021resilient}. In Remarks~\ref{remark_robustness} and \ref{remark_leaderless}, we discuss how the algorithm functions differently in the consensus problems under leaderless and leader-follower settings and the required network topologies for the algorithm to properly function.

\subsection{Threat Model}

We introduce our threat models extended from those studied in \cite{vaidya2012iterative,leblanc2013resilient,yuan2021resilient}.
\begin{definition}
	\textit{($f$-total/$f$-local set)}
	The set of adversary nodes $\mathcal{A}$ is said to be $f$-total
	if it contains at most $f$ nodes, i.e., $\left| \mathcal{A}\right| \leq f$.
	Similarly, it is said to be $f$-local (in $l$-hop neighbors)
	if any normal node $i$ has at most $f$ adversary nodes as its $l$-hop neighbors at any time $k$, i.e., $\left|\mathcal{N}_i^{l-}[k] \cap \mathcal{A}\right| \leq f, \forall i\in \mathcal{N}, \forall k$.
\end{definition}

\begin{definition}
	\textit{(Byzantine nodes)}
	An adversary node $i\in \mathcal{A}$ is said to be Byzantine
	if it arbitrarily modifies its own value and relayed values and sends different state values and relayed values to its neighbors at each step.
\end{definition}

Byzantine agents have been studied in numerous existing works and they are usually employed to characterize possible misbehaviors of adversarial agents in point-to-point networks \cite{Lynch,vaidya2012iterative,leblanc2013resilient,yuan2022asynchronous}. In contrast, there is the class of malicious agents, which are less adversarial than Byzantine agents as they are limited to send the identical false information to neighbors. Such agents form a suitable model for broadcast networks \cite{leblanc2013resilient, dibaji2017resilient} and wireless sensor networks \cite{iori2024resilient}.

As commonly done in the literature \cite{su2017reaching,Lynch,yuan2021resilient}, we assume that normal nodes have access to the neighbors' topology information and the bound on the number of adversaries.
\begin{assumption}
	Each node $i\in \mathcal{N}$ knows the value of $f$ and the topology information of its neighbors up to $l$ hops at each time $k$.
\end{assumption}

Moreover, to keep the problem tractable, we introduce the following assumption \cite{su2017reaching,yuan2021resilient}. It is merely introduced for ease of analysis. In fact, manipulating message paths can be easily detected and hence does not create problems. Related discussions can be found in \cite{su2017reaching,yuan2021resilient}.

\begin{assumption}
	Each node $i\in \mathcal{A}$ can manipulate its state $x_i[k]$ and the values in messages that they send or relay, but cannot change the path $P$ in such messages. 
\end{assumption}

\section{Jointly Robust Following Graphs}\label{sec_robustness}

In this section, we introduce a novel notion of jointly robust following graphs, which plays a key role in our resilient leader-follower consensus problem in time-varying networks.

\subsection{Jointly r-Reachable Followers with l Hops}

To establish a tight graph condition for our problem, we start with the definition of jointly reachable followers.
Under multi-hop communication, neighbors' values from the outside of the set to which each follower belongs may come from remote nodes and are not restricted to direct neighbors.

\begin{definition}\label{reachability}
	\textit{(Jointly reachable followers)}
	Consider the time-varying graph $\mathcal{G}[k] = (\mathcal{V},\mathcal{E}[k])$ with $l$-hop communication. For $r\in \mathbb{Z}_{>0}$ and a nonempty set $\mathcal{S}\subset \mathcal{W}$, we say that a node $i\in \mathcal{S}$ is a jointly $r$-reachable follower with $l$ hops in time interval $[k_t, k_{t+1})_{t\in\mathbb{Z}_{\geq 0}}$ if there exists a time $K_i \in [k_t, k_{t+1})$ such that
	\begin{equation*}
	|\mathcal{I}_{i, \mathcal{S}}[K_i]| \geq r, 
	\end{equation*}
	where $\mathcal{I}_{i, \mathcal{S}}[K_i]$ is the set of independent paths\footnote{Note that in these paths, only node $i$ is common.} to node $i$ of at most $l$ hops originating from nodes outside $\mathcal{S}$ at time $K_i$.
\end{definition}

Note that for node $i\in \mathcal{S}$ to satisfy $|\mathcal{I}_{i, \mathcal{S}}[K_i]| \geq r$ at time $K_i$, there should be at least $r$ source nodes outside $\mathcal{S}$ and an independent path of length at most $l$ hops from each of the $r$ source nodes to node $i$. Such source nodes may or may not include leader nodes.

\begin{figure}[t]
	\centering
	\subfigure[\scriptsize{ }]{
		\includegraphics[width=0.295\linewidth ]{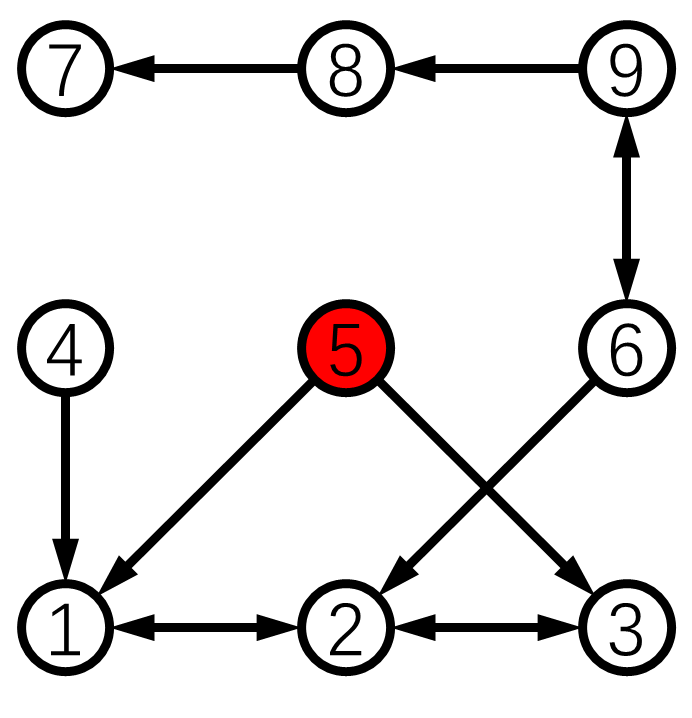}
	}
	\subfigure[\scriptsize{ }]{
		\includegraphics[width=0.295\linewidth ]{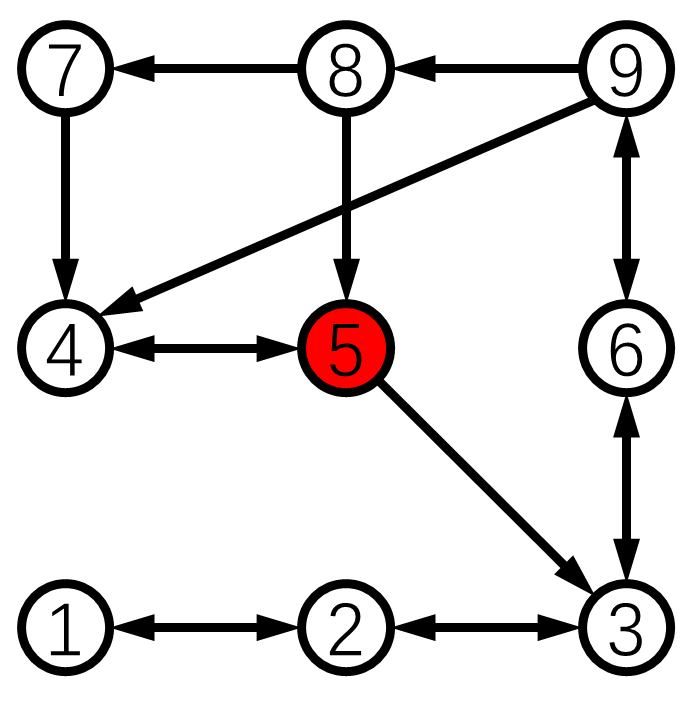}
	}
	\subfigure[\scriptsize{ }]{
		\includegraphics[width=0.295\linewidth ]{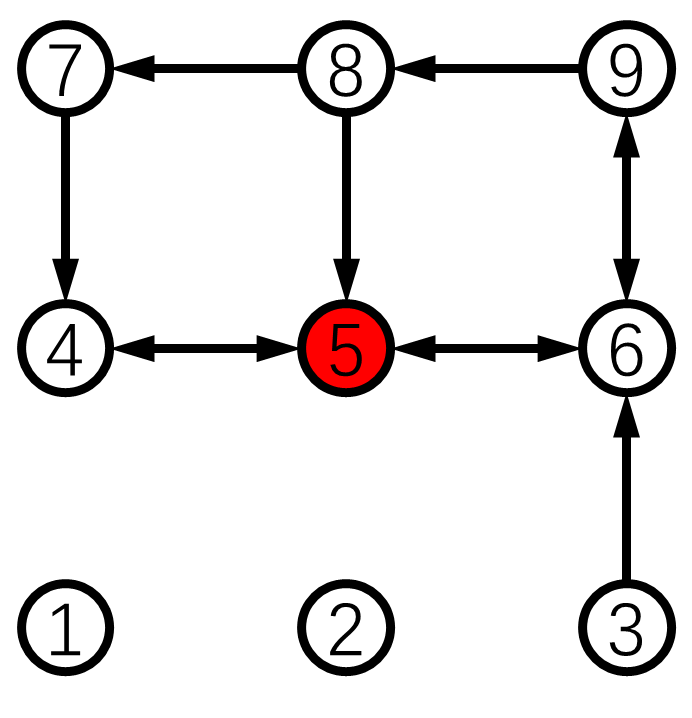}
	}
	\caption{The graph $\mathcal{G}[k]$ is not a jointly 2-robust following graph with 1 hop but is a jointly 2-robust following graph with 2 hops under the 1-local model. The set of leader agents $\mathcal{L}$ is \{7, 8, 9\}.}
	\label{9node2}
	\vspace*{-2.5mm}
\end{figure}

\begin{figure}[t]
	\centering
	\subfigure[\scriptsize{ }]{
		\includegraphics[width=0.295\linewidth ]{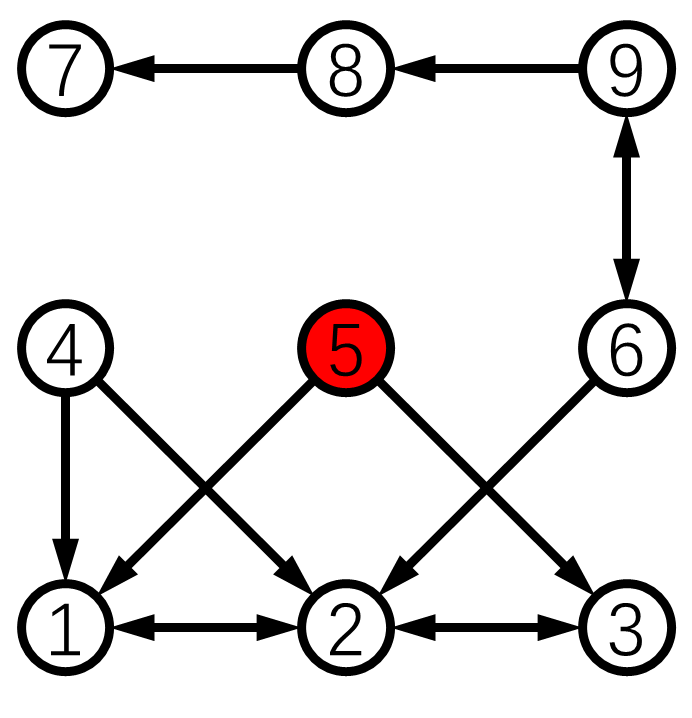}
	}
	\subfigure[\scriptsize{ }]{
		\includegraphics[width=0.295\linewidth ]{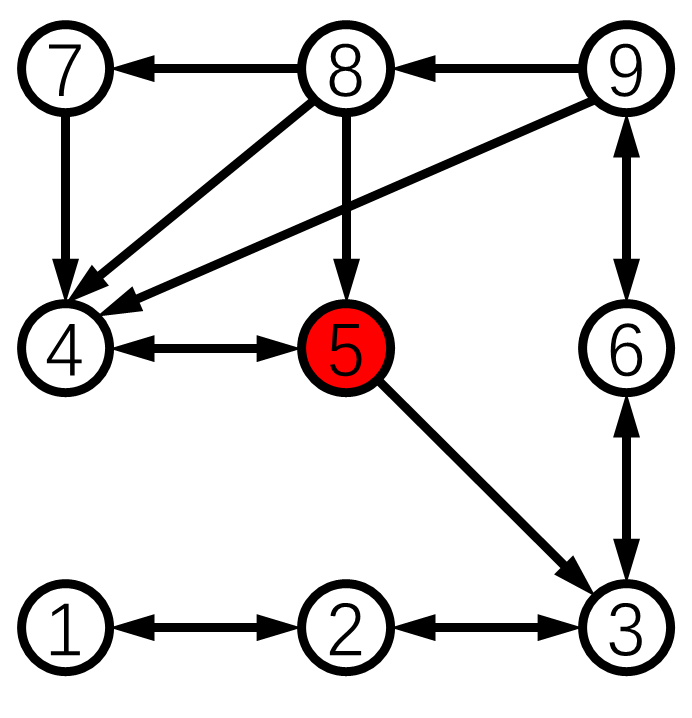}
	}
	\subfigure[\scriptsize{ }]{
		\includegraphics[width=0.295\linewidth ]{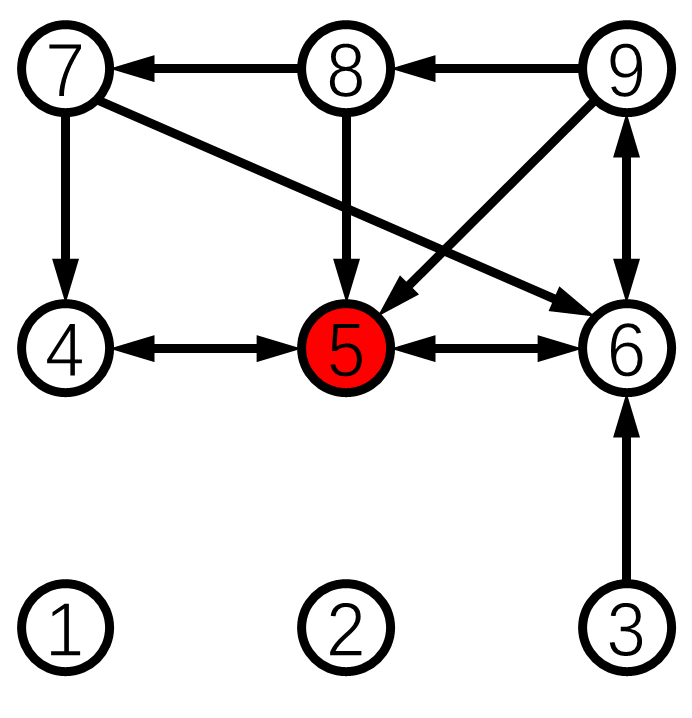}
	}
	\caption{The graph $\mathcal{G}[k]$ is a jointly 2-robust following graph with 1 hop under the 1-local model. The set of leader agents $\mathcal{L}$ is \{7, 8, 9\}.}
	\label{9node1}
	\vspace*{-2.5mm}
\end{figure}

\subsection{Jointly r-Robust Following Graphs with l Hops}
%Jointly $(r, s)$-robust with $l$ hops.

Now, we generalize the notion of jointly reachable followers to the entire graph and define jointly $r$-robust following graphs with $l$ hops as follows.

\begin{definition}\label{robust_following}
	\textit{(Jointly robust following graphs)}
	Consider the time-varying digraph $\mathcal{G}[k] = (\mathcal{V},\mathcal{E}[k])$ with the set of leaders $\mathcal{L}\subset \mathcal{V}$. Let $\mathcal{F} \subset \mathcal{V}$ and denote by $\mathcal{G}_{\mathcal{H}}[k]$ the subgraph of $\mathcal{G}[k]$ induced by node set $\mathcal{H}=\mathcal{V}\setminus\mathcal{F}$.
	Graph $\mathcal{G}[k]$ is said to be a jointly $r$-robust following graph with $l$ hops (under the $f$-local model) if for any $f$-local set $\mathcal{F}$, the subgraph $\mathcal{G}_{\mathcal{H}}[k]$ satisfies that there exists an infinite sequence of uniformly bounded time intervals (ISUBTI) $\{[k_t,k_{t+1})\}_{t\in\mathbb{Z}_{\geq 0}}$ with $ k_t< k_{t+1}$ and $k_0=0$ such that in each time interval, for every nonempty subset $\mathcal{S}\subseteq \mathcal{H}\setminus \mathcal{L}$, the following condition holds:
	\begin{align*}
		 | \mathcal{Z}_{\mathcal{S}}^{r}[k_t, k_{t+1}) | \geq 1  , \textup{where}
	\end{align*}
	$\mathcal{Z}_{\mathcal{S}}^{r}[k_t, k_{t+1})=\{i\in \mathcal{S} :  \exists K_i \in [k_t, k_{t+1}) \medspace \textup{s.t.} \medspace 
	|\mathcal{I}_{i, \mathcal{S}}[K_i]| \geq r \}$.
\end{definition}
\vspace{2mm}

%We note that the notion of joint $(r,s)$-robustness with 1 hop is equivalent to the notion of joint $(r,s)$-robustness in \cite{wen2023joint}. Moreover, when $l=1$, the consensus results for the malicious model in this paper coincide with the corresponding results in \cite{wen2023joint}.

\begin{figure}[t]
	\centering
	\subfigure[\scriptsize{ }]{
		\includegraphics[width=0.57\linewidth ]{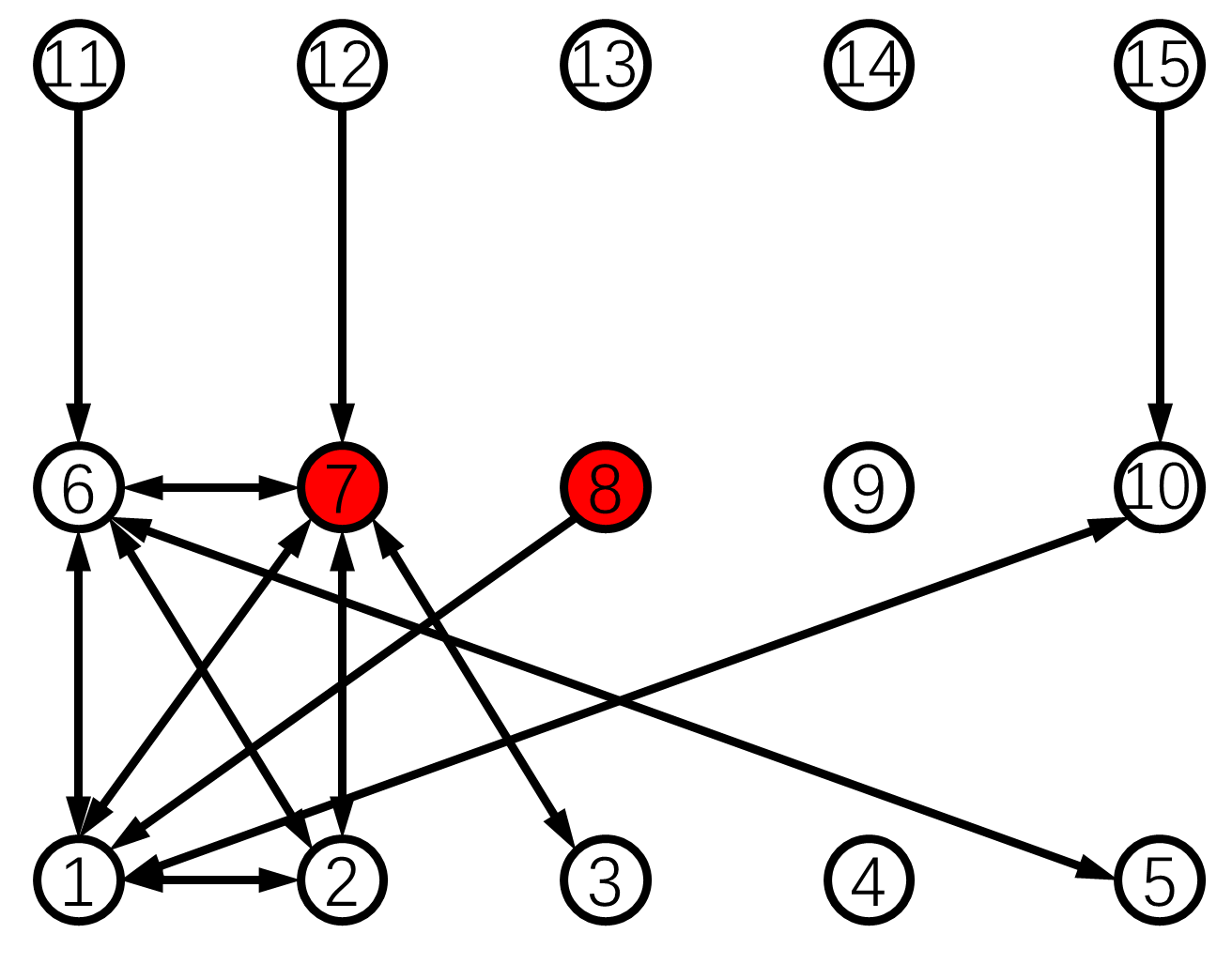}
	}
	\subfigure[\scriptsize{ }]{
		\includegraphics[width=0.57\linewidth ]{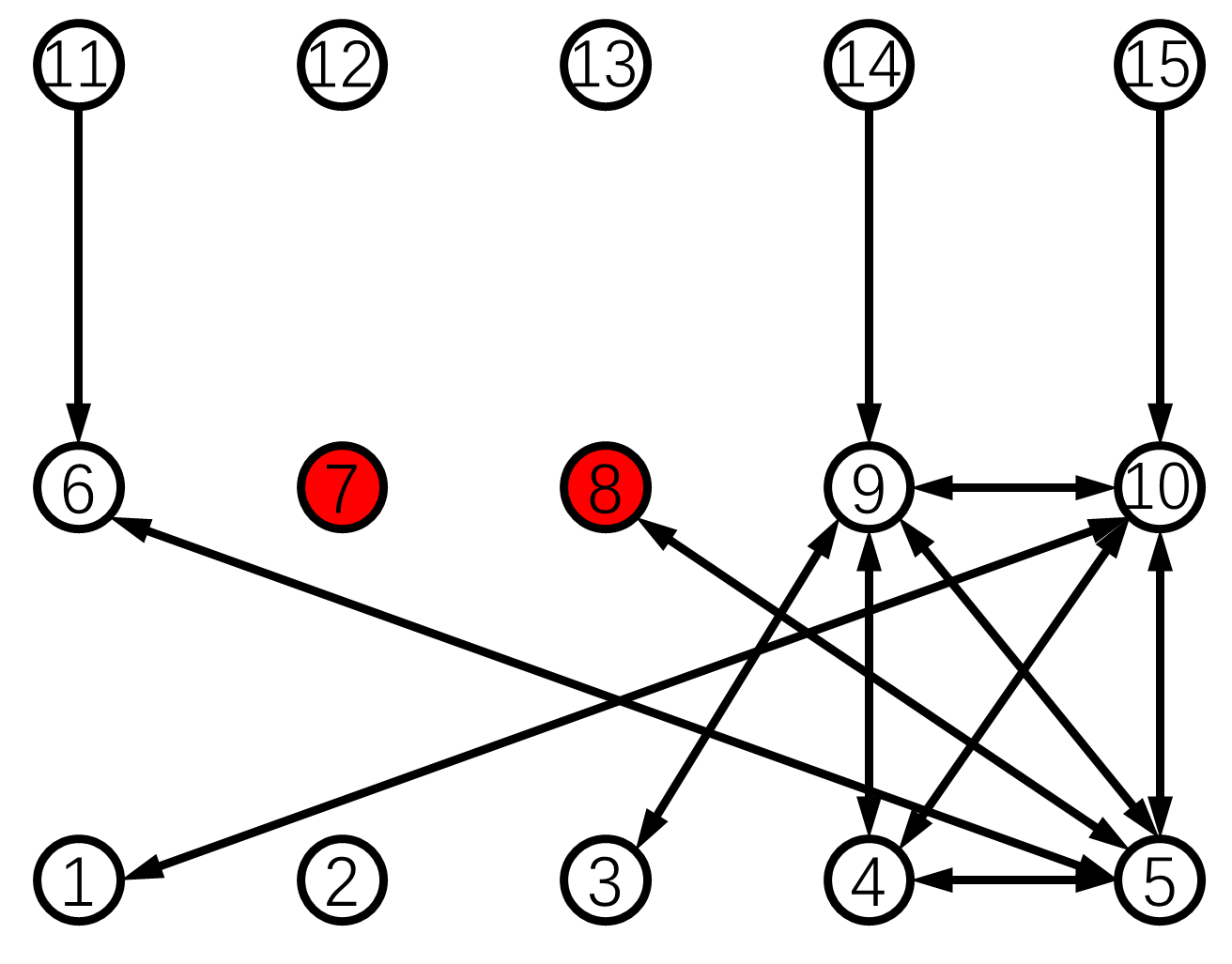}
	}
	\subfigure[\scriptsize{ }]{
		\includegraphics[width=0.57\linewidth ]{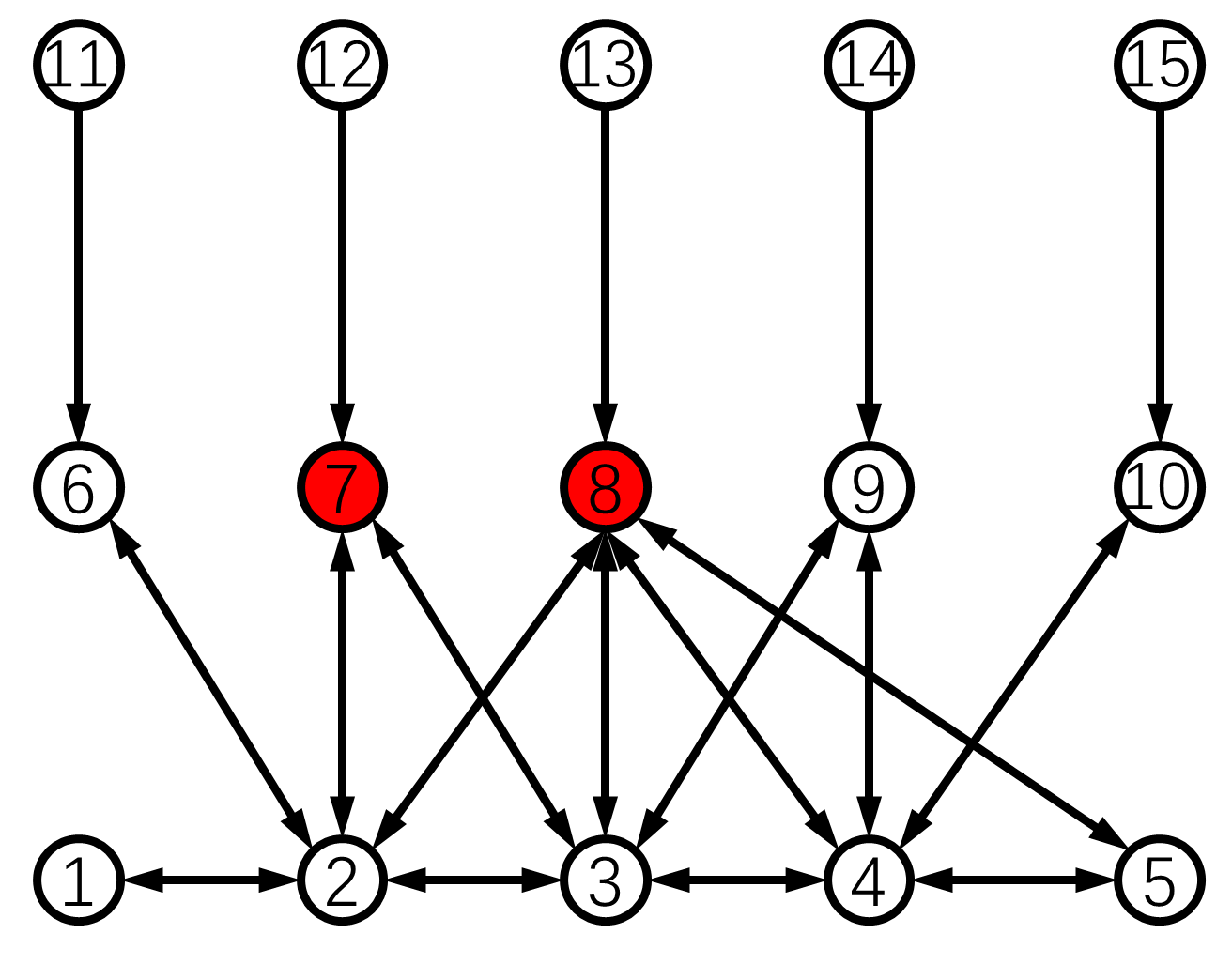}
	}
	\caption{The graph $\mathcal{G}[k]$ is not a jointly 3-robust following graph with 1 hop but is a jointly 3-robust following graph with 3 hops under the 2-local model. The set of leader agents $\mathcal{L}$ is \{11, 12, 13, 14, 15\}.}
	\label{15node}
	\vspace*{-2.5mm}
\end{figure}

We have reported in \cite{yuan2021resilient} that for static networks, graph robustness guaranteeing resilient consensus increases as the relay range $l$ increases. We emphasize that multi-hop relays can also improve the network robustness for our resilient leader-follower consensus problem.
We illustrate this idea using the time-varying networks in Figs.~\ref{9node2}--\ref{15node}, where the communication topology of each network switches among the three graphs. Here, the order of the graphs is arbitrary, but each graph should appear at least once in each interval $[k_t, k_{t+1})$. Note that when $l=1$, counting independent paths is equivalent to counting in-neighbors.

\begin{example}\label{discussion9node}
	Consider the graph $\mathcal{G}[k]$ in Fig.~\ref{9node2}.
	It should be noted that the notion of jointly $r$-robust following graphs with $l$ hops depends on the choice of set $\mathcal{F}$.
	We claim that this graph is not a jointly $2$-robust following graph with $1$ hop under the $1$-local model. For instance, after removing the node set $\mathcal{F}=\{5\}$, the remaining graph $\mathcal{G}_{\mathcal{H}}[k]$ does not satisfy the condition in Definition~\ref{robust_following}.
	The reason is that in the follower set $\mathcal{S}=\{1,2,3,6\}$, none of the nodes has 2 in-neighbors outside $\mathcal{S}$, i.e., $\mathcal{Z}_{\mathcal{S}}^{2}[k_t, k_{t+1})=\emptyset$ when $l=1$. 
	In fact, for this network to become a jointly $2$-robust following graph with $1$ hop, four more edges are needed as depicted in Fig.~\ref{9node1}. Alternatively, we can also increase the network robustness by increasing the relay range. For example, when $l=2$, for node sets $\mathcal{F}=\{5\}$ and $\mathcal{S}=\{1,2,3,6\}$, node 2 has 2 independent two-hop paths originating from nodes outside $\mathcal{S}$, i.e., $\mathcal{Z}_{\mathcal{S}}^{2}[k_t, k_{t+1})\neq \emptyset$.
	Moreover, one can verify all the combinations of node subsets and conclude that this graph is a jointly $2$-robust following graph with $2$ hops under the $1$-local model.
\end{example}

\begin{example}\label{discussion15node}
	Consider the larger graph in Fig.~\ref{15node}.
	It is not a jointly $3$-robust following graph with $1$ hop under the $2$-local model.
	Observe that after removing the node set $\mathcal{F}=\{7,8\}$, the remaining graph $\mathcal{G}_{\mathcal{H}}[k]$ does not satisfy the condition in Definition~\ref{robust_following}. Specifically, in the follower set $\mathcal{S}=\mathcal{W}\setminus\mathcal{F}$, none of the nodes has 3 in-neighbors outside $\mathcal{S}$, i.e., $\mathcal{Z}_{\mathcal{S}}^{3}[k_t, k_{t+1})=\emptyset$ when $l=1$. 
	However, after verifying all the combinations of node subsets, we can conclude that this graph is a jointly $3$-robust following graph with $3$ hops under the $2$-local model.
\end{example}

Then, we introduce a simpler version of Definition~\ref{robust_following} for static networks as follows.

\begin{definition}\label{robust_following_static}
	\textit{(Robust following graphs)}
	If a static graph $\mathcal{G}$ is a jointly $(f+1)$-robust following graph with $l$ hops where $k_{t+1}-k_t=1,  \forall [k_t, k_{t+1})$, we simply say that $\mathcal{G}$ is an $(f+1)$-robust following graph with $l$ hops.
\end{definition}

\begin{remark}\label{remark_robustness}
	There is an intuitive method for achieving leader-follower consensus using the leaderless consensus results. That is, adding sufficient number of leaders and corresponding edges to a follower subgraph satisfying the condition of strict robustness with $l$ hops. This condition is necessary and sufficient for achieving leaderless resilient consensus under the Byzantine model \cite{yuan2022asynchronous}. However, we will prove in Theorem~\ref{theorem_firstorder} that our condition of robust following graphs is necessary and sufficient for resilient leader-follower consensus. Hence, our new condition is tighter for the leader-follower case compared to the one based on leaderless results. This can also be observed from the example graph in Fig~\ref{9node1}. Consider the union of the graphs in Fig~\ref{9node1}. Its subgraph of followers does not satisfy the strict robustness condition\footnote{If graph $\mathcal{G}$ is $(f+1)$-strictly robust with $l$ hops, then its minimum in-degree must be no less than $2f+1$.} since node 4 has only one incoming edge from neighbor 5. Besides, the analysis for resilient leader-follower consensus is very different from the one for leaderless consensus \cite{yuan2022asynchronous}, as we will see in Theorem~\ref{theorem_firstorder} and Remark~\ref{remark_leaderless}.
\end{remark}

%\begin{definition}\label{robust_following_static}
%	\textit{(Robust following graphs)}
%	Consider the static digraph $\mathcal{G} = (\mathcal{V},\mathcal{E})$ with the set of leaders $\mathcal{L}\subset \mathcal{V}$. Let $\mathcal{F} \subset \mathcal{V}$ and denote the subgraph of $\mathcal{G} = (\mathcal{V},\mathcal{E})$ induced by node set $\mathcal{H}=\mathcal{V}\setminus\mathcal{F}$ as $\mathcal{G}_{\mathcal{H}}$.
%	Graph $\mathcal{G}$ is said to be an $r$-robust following graph with $l$ hops (under the $f$-local model) if for any $f$-local set $\mathcal{F}$, the subgraph $\mathcal{G}_{\mathcal{H}}$ satisfies that for every nonempty subset $\mathcal{S}\subseteq \mathcal{H}\setminus \mathcal{L}$, the following condition holds:
%	\begin{align*}
%    | \mathcal{Z}_{\mathcal{S}}^{r} | =| \{i\in \mathcal{S} : |\mathcal{I}_{i, \mathcal{S}}| \geq r \} | \geq 1 .
%	\end{align*}
%\end{definition}
%\vspace{2mm}

Next, we introduce the notion of normal network consisting of only normal nodes.

\begin{definition}
	\textit{(Normal network)}
	For a network $\mathcal{G}[k] = (\mathcal{V}, \mathcal{E}[k])$, define the normal
	network of $\mathcal{G}[k]$, denoted by $\mathcal{G}_{\mathcal{N}}[k]$, as the network induced by the normal nodes, i.e., $\mathcal{G}_{\mathcal{N}}[k] = (\mathcal{N},\mathcal{E}_{\mathcal{N}}[k] )$, where $\mathcal{E}_{\mathcal{N}}[k]$ is the set of directed edges among normal nodes at each time $k$.
\end{definition}

It is worth noting that we could have a tighter graph condition defined on the normal network for our resilient leader-follower consensus problem. Such a condition is that the normal network $\mathcal{G}_{\mathcal{N}}[k]$ satisfies the property presented in Definition~\ref{robust_following} for $\mathcal{G}_{\mathcal{H}}[k]$ there. In fact, if $\mathcal{G}[k]$ is a jointly $r$-robust following graph with $l$ hops under the $f$-local model, then the normal network $\mathcal{G}_{\mathcal{N}}[k]$ is guaranteed to satisfy the abovementioned property. However, the graph condition on the normal network cannot be verified prior to the actual deployment of the proposed algorithm. Therefore, similar to the works \cite{vaidya2012iterative,su2017reaching,yuan2021resilient}, we define our condition on the original network topology $\mathcal{G}[k]$.

\section{Resilient Leader-Follower Consensus in First-Order MASs}

In this section, we see how the MW-MSR algorithm guarantees resilient leader-follower consensus in time-varying directed networks, where each normal follower utilizes the scalar values received from all of its $l$-hop neighbors synchronously to update its next value.

\subsection{Convergence Analysis}

Regarding the leader-follower consensus error, we define three variables. These are,  respectively, the maximum value, the minimum value of normal nodes in $\mathcal{N}$ at time $k$, and their differences given by
\begin{align}\label{consensus_error}
	\overline{x}[k]&= \max_{i\in \mathcal{W}^\mathcal{N}, \medspace d\in \mathcal{L}^\mathcal{N}} \medspace \{ x_i[k] , x_d[k] \}, \nonumber \\[1mm]
	\underline{x}[k]&= \min_{i\in \mathcal{W}^\mathcal{N}, \medspace d\in \mathcal{L}^\mathcal{N}} \medspace \{ x_i[k] , x_d[k] \}, \nonumber \\[1mm]
	V[k] &= \overline{x}[k]-\underline{x}[k]. 
\end{align}
%\vspace{-2mm}

We state a lemma describing behaviors of normal followers.

\begin{lemma}\label{lemma_safety}
	Consider the time-varying network $\mathcal{G}[k] = (\mathcal{V},\mathcal{E}[k])$ with $l$-hop communication, where each normal follower node $i\in \mathcal{W}^\mathcal{N}$ updates its value according to the MW-MSR algorithm with parameter $f$. Under the $f$-local adversarial set $\mathcal{A}$ and the assumption that $r [k]$ is constant $\forall k\in [k_1, k_e)$, the following statements hold $\forall k\in [k_1, k_e)$:
	\begin{enumerate}
		\item $x_i [k] \in [\underline{x}[k_1], \overline{x}[k_1]], \forall i\in \mathcal{N}$,
		\item  $\big[ \underline{x}[k+1], \overline{x}[k+1] \big] \subset \big[ \underline{x}[k], \overline{x}[k] \big]$.
	\end{enumerate}
\end{lemma}
\vspace{2mm}

\begin{proof}
Since $\forall d\in \mathcal{L}^\mathcal{N}$, $x_d [k]= r [k]$ is constant $\forall k\in [k_1, k_e)$, we have $x_d [k] \in [\underline{x}[k_1], \overline{x}[k_1]], \forall d\in \mathcal{L}^\mathcal{N}$. Then we consider any follower node $i\in \mathcal{W}^\mathcal{N}$.
For node $i$, the values used in \eqref{msrupdate} always lie within the interval $\big[ \underline{x}[k], \overline{x}[k] \big]$ for $k\geq k_1$.
This can be seen from the fact that at each time $k$, node $i$ removes the possibly manipulated values from at most $f$ nodes within $l$ hops in step 2 of Algorithm 1.
Hence, the update rule \eqref{msrupdate} uses a convex combination of the values in $\big[ \underline{x}[k], \overline{x}[k]\big]$ and it holds that $x_i[k+1] \in  \big[ \underline{x}[k], \overline{x}[k] \big], \forall i\in \mathcal{W}^\mathcal{N}, \forall k\in [k_1, k_e) $. Thus, we have
$\big[ \underline{x}[k+1], \overline{x}[k+1] \big] \subset \big[ \underline{x}[k], \overline{x}[k] \big] \subset \cdots \subset  \big[ \underline{x}[k_1], \overline{x}[k_1] \big], \forall i\in \mathcal{N}, \forall k\in [k_1, k_e).$
\end{proof}

The following theorem is the first main contribution of this paper.
Here, we characterize a necessary and sufficient condition for time-varying networks using the MW-MSR algorithm to achieve resilient leader-follower consensus under the $f$-local Byzantine model. 

\begin{theorem}\label{theorem_firstorder}
Consider the time-varying network $\mathcal{G}[k] = (\mathcal{V},\mathcal{E}[k])$ with $l$-hop communication, where each normal follower node $i\in \mathcal{W}^\mathcal{N}$ updates its value according to the MW-MSR algorithm with parameter $f$. Under the $f$-local adversarial set $\mathcal{A}$ and the assumption that $r [k]$ is constant $\forall k\in [k_1, \infty)$, resilient leader-follower consensus is achieved if and only if $\mathcal{G}[k]$ is a jointly $(f+1)$-robust following graph with $l$ hops.
\end{theorem}

\begin{proof}
	\textit{(Necessity)} If $\mathcal{G}[k]$ is not a jointly $(f+1)$-robust following graph with $l$ hops, then by Definition~\ref{robust_following}, 
	there exists an $f$-local set $\mathcal{F}$ such that $\mathcal{G}_{\mathcal{H}}[k]$ does not satisfy the condition there. Suppose that $\mathcal{F}$ is exactly the set of Byzantine agents $\mathcal{A}$. Then, in the normal network $\mathcal{G}_{\mathcal{N}}[k] = (\mathcal{N},\mathcal{E}_{\mathcal{N}}[k] )$, there must be a nonempty subset $\mathcal{S}\subseteq \mathcal{N}\setminus \mathcal{L}$ such that $\mathcal{Z}_{\mathcal{S}}^{f+1}[\hat{k}, \infty) = \emptyset $ for some $\hat{k} > k_1$. It further means that
	\begin{align}\label{atmostf}
		|\mathcal{I}_{i, \mathcal{S}}[k]| \leq f , \medspace \forall k \in [\hat{k}, \infty), \medspace \forall i\in \mathcal{S}.
	\end{align}

	Suppose that $x_i[\hat{k}]=a , \medspace \forall i\in \mathcal{S}$, and $x_j[\hat{k}]=r[k_1], \medspace \forall j\in \mathcal{N}\setminus \mathcal{S}$, where $a< r[k_1]$ is a constant. Moreover, suppose that all Byzantine nodes send $a$ and $r[k_1]$ to the nodes in $\mathcal{S}$ and $ \mathcal{N}\setminus \mathcal{S}$, respectively. 
	For normal node $i \in \mathcal{S}$, \eqref{atmostf} indicates that the cardinality of the MMC of the values larger than its own value (i.e., values from the normal nodes outside of $\mathcal{S}$) is at most $f$. These values are disregarded by the MW-MSR algorithm. Moreover, since the Byzantine nodes send $a$ to node $i$, it will use these values. Thus, node $i$ will keep a constant value $a$ after time $\hat{k}$ and resilient leader-follower consensus cannot be achieved.
	
	\textit{(Sufficiency)} 
	Let $\underline{\epsilon}_0=r[k_1] - \underline{x}[k_1]$ and $\overline{\epsilon}_0=\overline{x}[k_1] - r[k_1] $.
	Recall that $\alpha\in (0,1)$ is the lower bound for the coefficients in \eqref{msrupdate}. Denote by $K$ the maximum length of time intervals $\{[k_t,k_{t+1})\}_{t\in\mathbb{Z}_{\geq 0}}$.
	For $\gamma= 0, 1,2,\dots,(|\mathcal{W}^\mathcal{N}|+1)K$, define $\overline{\epsilon}_{\gamma}$ and $\underline{\epsilon}_{\gamma}$ as
	\begin{align*}
		\overline{\epsilon}_{\gamma} = \alpha^\gamma \overline{\epsilon}_0 , \medspace\medspace
		\underline{\epsilon}_{\gamma} = \alpha^\gamma \underline{\epsilon}_0 .
	\end{align*}
	So we have $0 < \overline{\epsilon}_{\gamma+1} < \overline{\epsilon}_{\gamma}\leq \overline{\epsilon}_0$ and $0 < \underline{\epsilon}_{\gamma+1} < \underline{\epsilon}_{\gamma}\leq \underline{\epsilon}_0$, for all $\gamma$.

	For any time $k \geq k_1$ and any $\gamma$, we define the following sets:
	\begin{align*}
		\mathcal{Z}_1(k,\overline{\epsilon}_{\gamma})&=\{i\in \mathcal{W}^\mathcal{N}  :  x_i[k]>\overline{x}[k_1] - \overline{\epsilon}_{\gamma} \},\\[1mm]
		\mathcal{Z}_2(k,\underline{\epsilon}_{\gamma})&=\{i\in \mathcal{W}^\mathcal{N} : x_i[k]<\underline{x}[k_1] + \underline{\epsilon}_{\gamma} \},\\[1mm]
		\mathcal{U}(k,\overline{\epsilon}_{\gamma},\underline{\epsilon}_{\gamma})&=\mathcal{Z}_1(k,\overline{\epsilon}_{\gamma}) \cup \mathcal{Z}_2(k,\underline{\epsilon}_{\gamma}), \\[1mm]
		\overline{\mathcal{U}}(k,\overline{\epsilon}_{\gamma},\underline{\epsilon}_{\gamma})&=\mathcal{W}^\mathcal{N}  \setminus \mathcal{U}(k,\overline{\epsilon}_{\gamma},\underline{\epsilon}_{\gamma}).
%		\mathcal{X}(k,\overline{\epsilon}_{\gamma},\underline{\epsilon}_{\gamma})&=\{i\in \mathcal{U}(k,\overline{\epsilon}_{\gamma},\underline{\epsilon}_{\gamma})  : \\[1mm]
%		&\hspace*{0.5cm}\mbox{} | \mathcal{I}_{i\in \mathcal{U}(k,\overline{\epsilon}_{\gamma},\underline{\epsilon}_{\gamma})}[k]| \geq f+1  \}.
	\end{align*}

	We will show in three steps that $|\mathcal{U}(k,\overline{\epsilon}_{\gamma},\underline{\epsilon}_{\gamma})|$ decreases over an appropriate sequence of $\gamma$. Once this set is empty after a period, $V[k]$ will decrease. We will further prove the consensus result when $k\to \infty$.

	\vspace{2mm}
	\noindent \textbf{Step 1:}
	
	Since $\mathcal{G}[k]$ is a jointly $(f+1)$-robust following graph with $l$ hops and the adversarial set $\mathcal{A}$ is an $f$-local set, for the normal network $\mathcal{G}_{\mathcal{N}}[k] = (\mathcal{N},\mathcal{E}_{\mathcal{N}}[k] )$, there exists an ISUBTI $\{[k_t,k_{t+1})\}_{t\in\mathbb{Z}_{\geq 0}}$ such that in each time interval, for every nonempty subset $\mathcal{S}\subseteq \mathcal{N}\setminus \mathcal{L}$, the following condition holds:
	\begin{align*}
		& | \mathcal{Z}_{\mathcal{S}}^{f+1}[k_t, k_{t+1}) | \\
		&=| \{i\in \mathcal{S} :  \exists K_i \in [k_t, k_{t+1}) \medspace \textup{s.t.} \medspace 
		|\mathcal{I}_{i, \mathcal{S}}[K_i]| \geq f+1 \} | \geq 1  .
	\end{align*}
	Hence, we can always find a nonempty subset $\mathcal{W}_1 \subseteq \mathcal{W}^\mathcal{N} = \mathcal{N}\setminus \mathcal{L}$ such that $\forall i_1 \in \mathcal{W}_1$, there is $K_{i_1} \in [k_1, k_2)$ such that $|\mathcal{I}_{i_1, \mathcal{W}^\mathcal{N}}[K_{i_1}]| \geq f+1$ for the first time in $[k_1, k_2)$. Therefore, any node $ i_1 \in \mathcal{W}_1$ will use at least one value from normal leaders in $\mathcal{L}^\mathcal{N}$ after applying the MW-MSR algorithm. This can be seen from step 2 of Algorithm~1, where node $i_1$ can only remove the values from at most $f$ nodes sharing the same value. Notice that in \eqref{msrupdate}, each $a_i[k]$ is lower bounded by $\alpha$. Moreover, from Lemma~\ref{lemma_safety}, $x_i [k] \in [\underline{x}[k_1], \overline{x}[k_1]], \forall i\in \mathcal{N}, \forall k\in [k_1, \infty)$, 
	Hence, we obtain the following bound:
	\begin{align}\label{lower}
		x_{i_1}[K_{i_1}+1] &\geq (1-\alpha)\underline{x}[k_1]+\alpha r[k_1] \nonumber\\
		&\geq \underline{x}[k_1]+\alpha \underline{\epsilon}_0.
	\end{align}
	Extending these bounds to time $k_2$ yields
	\begin{align*}
		x_{i_1}[K_{i_1}+2] &\geq (1-\alpha)\underline{x}[k_1]+\alpha x_{i_1}[k_1+1] \\
		&\geq (1-\alpha)\underline{x}[k_1]+\alpha (\underline{x}[k_1]+\alpha \underline{\epsilon}_0) \\
		&\geq \underline{x}[k_1]+\alpha^2 \underline{\epsilon}_0, \\
		x_{i_1}[K_{i_1}+3] &\geq (1-\alpha)\underline{x}[k_1]+\alpha (\underline{x}[k_1]+\alpha^2 \underline{\epsilon}_0) \\
		&\geq \underline{x}[k_1]+\alpha^3 \underline{\epsilon}_0, \\
		&  \hspace*{0.19cm}\mbox{} \vdots   \\
		x_{i_1}[k_2] &\geq (1-\alpha)\underline{x}[k_1]+\alpha (\underline{x}[k_1]+\alpha^{k_2-K_{i_1}-1} \underline{\epsilon}_0) \\
		&\geq \underline{x}[k_1]+\alpha^{k_2-K_{i_1}} \underline{\epsilon}_0 \\
		&\geq \underline{x}[k_1]+\alpha^{k_2-k_1} \underline{\epsilon}_0.
	\end{align*}
	Using similar arguments, we can establish the upper bounds for node $i_1$ as
	\begin{align*}
		x_{i_1}[K_{i_1}+1] &\leq (1-\alpha)\overline{x}[k_1]+\alpha r[k_1], \\
		&  \hspace*{0.19cm}\mbox{} \vdots   \\
		x_{i_1}[k_2] &\leq (1-\alpha)\overline{x}[k_1]+\alpha (\overline{x}[k_1]+\alpha^{k_2-K_{i_1}-1} \overline{\epsilon}_0) \\
		&\leq \overline{x}[k_1]-\alpha^{k_2-k_1} \overline{\epsilon}_0.
	\end{align*}
	Hence, we obtain $x_{i_1}[k_2]\in [\underline{x}[k_1]+\alpha^{k_2-k_1} \underline{\epsilon}_0, \overline{x}[k_1]-\alpha^{k_2-k_1} \overline{\epsilon}_0] = \overline{\mathcal{U}}(k_2,\overline{\epsilon}_{k_2-k_1},\underline{\epsilon}_{k_2-k_1}), \forall i_1 \in \mathcal{W}_1$.
	
	\vspace{2mm}
	\noindent \textbf{Step 2:}
	
	We show next that $| \mathcal{U}(k_3,\overline{\epsilon}_{k_3-k_1},\underline{\epsilon}_{k_3-k_1}) |< |\mathcal{W}^\mathcal{N}|$.
	Observe that $\mathcal{W}_1 \subseteq \overline{\mathcal{U}}(k_2,\overline{\epsilon}_{k_2-k_1},\underline{\epsilon}_{k_2-k_1})$, and hence, $\overline{\mathcal{U}}(k_2,\overline{\epsilon}_{k_2-k_1},\underline{\epsilon}_{k_2-k_1})$ is nonempty. Since each node in $\mathcal{W}^\mathcal{N}$ always uses its own state in \eqref{msrupdate}, lower bounds on the values of any node $i_2 \in \overline{\mathcal{U}}(k_2,\overline{\epsilon}_{k_2-k_1},\underline{\epsilon}_{k_2-k_1})$ can be established as
	\begin{align*}
		x_{i_2}[k_2+1] &\geq (1-\alpha)\underline{x}[k_1]+\alpha x_{i_1}[k_2] \\
		&\geq (1-\alpha)\underline{x}[k_1]+\alpha (\underline{x}[k_1]+\alpha^{k_2-k_1} \underline{\epsilon}_0) \\
		&\geq \underline{x}[k_1]+\alpha^{k_2+1-k_1} \underline{\epsilon}_0, \\
		x_{i_2}[k_2+2] &\geq \underline{x}[k_1]+\alpha^{k_2+2-k_1} \underline{\epsilon}_0, \\
		&  \hspace*{0.19cm}\mbox{} \vdots   \\
		x_{i_2}[k_3] &\geq \underline{x}[k_1]+\alpha^{k_3-k_1} \underline{\epsilon}_0.
	\end{align*}
	Similarly, the following upper bounds hold for node $i_2$:
	\begin{align*}
		x_{i_2}[k_2+1] &\leq (1-\alpha)\overline{x}[k_1]+\alpha x_{i_1}[k_2] \\
		&\leq (1-\alpha)\overline{x}[k_1]+\alpha (\overline{x}[k_1]-\alpha^{k_2-k_1} \overline{\epsilon}_0) \\
		&\leq \overline{x}[k_1]-\alpha^{k_2+1-k_1} \overline{\epsilon}_0, \\
		x_{i_2}[k_2+2] &\leq \overline{x}[k_1]-\alpha^{k_2+2-k_1} \overline{\epsilon}_0, \\
		&  \hspace*{0.19cm}\mbox{} \vdots   \\
		x_{i_2}[k_3] &\leq \overline{x}[k_1]-\alpha^{k_3-k_1} \overline{\epsilon}_0.
	\end{align*}
	Hence, $\overline{\mathcal{U}}(k_2,\overline{\epsilon}_{k_2-k_1},\underline{\epsilon}_{k_2-k_1}) \cap \mathcal{U}(k_3,\overline{\epsilon}_{k_3-k_1},\underline{\epsilon}_{k_3-k_1}) = \emptyset$, indicating that $| \mathcal{U}(k_3,\overline{\epsilon}_{k_3-k_1},\underline{\epsilon}_{k_3-k_1}) |< |\mathcal{W}^\mathcal{N}|$.

	\vspace{2mm}
	\noindent \textbf{Step 3:}
	
	In the rest of the proof, we will show that $\forall t\geq 3$, it holds
	\begin{align}\label{shrinking}
		| \mathcal{U}(k_t,&\overline{\epsilon}_{k_t-k_1},\underline{\epsilon}_{k_t-k_1}) | \nonumber   \\ &>  | \mathcal{U}(k_{t+1}, \overline{\epsilon}_{k_{t+1}-k_1},\underline{\epsilon}_{k_{t+1}-k_1}) |,
	\end{align}
	until $ \mathcal{U}(k_{t},\overline{\epsilon}_{k_{t}-k_1},\underline{\epsilon}_{k_{t}-k_1})$ becomes empty.

	Since $\mathcal{G}[k]$ is a jointly $(f+1)$-robust following graph with $l$ hops and the adversarial set $\mathcal{A}$ is an $f$-local set, in every time interval $[k_t, k_{t+1})$, there exists a nonempty subset $\mathcal{X}_t \subseteq \mathcal{U}(k_t,\overline{\epsilon}_{k_t-k_1},\underline{\epsilon}_{k_t-k_1})$
	so that $\forall i_3 \in \mathcal{X}_t$, it holds that $\exists K_{i_3} \in [k_t, k_{t+1})  \medspace \textup{s.t.} \medspace 
	|\mathcal{I}_{i_3, \mathcal{U}(k_t,\overline{\epsilon}_{k_t-k_1},\underline{\epsilon}_{k_t-k_1})}[K_{i_3}]| \geq f+1$ for the first time in $[k_t, k_{t+1})$.
	Observe that $\forall i_3 \in \mathcal{X}_t$, either one of the following holds:
	\begin{align*}
		x_{i_3}[k_t] &>\overline{x}[k_1] - \overline{\epsilon}_{k_t-k_1},   \\ 
		x_{i_3}[k_t] &<\underline{x}[k_1] + \underline{\epsilon}_{k_t-k_1}  .
	\end{align*}
	Therefore, any node $ i_3 \in \mathcal{X}_t$ will use at least one normal neighbor’s value from the interval $[\underline{x}[k_1]+\alpha^{k_t-k_1} \underline{\epsilon}_0, \overline{x}[k_1]-\alpha^{k_t-k_1} \overline{\epsilon}_0]$ in \eqref{msrupdate} when it applies the MW-MSR algorithm. This can be seen from the facts that node $i_3$ has at least $f+1$ normal neighbors within $l$ hops outside the set $\mathcal{U}(k_t,\overline{\epsilon}_{k_t-k_1},\underline{\epsilon}_{k_t-k_1})$ and node $i_3$ can only remove the values from at most $f$ nodes having values smaller or larger than its own value.
	Hence, we obtain the following bounds $\forall i_3 \in \mathcal{X}_t$:
	\begin{align}\label{lower_pk}
		x_{i_3}[K_{i_3}+1] 
		&\geq (1-\alpha)\underline{x}[k_1]+\alpha (\underline{x}[k_1]+\alpha^{k_t-k_1} \underline{\epsilon}_0) \nonumber \\
		&\geq \underline{x}[k_1]+\alpha^{1+k_t-k_1} \underline{\epsilon}_0, \nonumber \\
		x_{i_3}[K_{i_3}+2] &\geq \underline{x}[k_1]+\alpha^{2+k_t-k_1} \underline{\epsilon}_0, \nonumber \\
		&  \hspace*{0.19cm}\mbox{} \vdots  \nonumber \\
		x_{i_3}[k_{t+1}] &\geq \underline{x}[k_1]+\alpha^{k_{t+1}-K_{i_3}+k_t-k_1} \underline{\epsilon}_0 \nonumber\\
		&\geq \underline{x}[k_1]+\alpha^{k_{t+1}-k_1} \underline{\epsilon}_0.
	\end{align}
	Similarly, the following upper bounds can be established $\forall i_3 \in \mathcal{X}_t$:
	\begin{align}\label{upper_pk}
		x_{i_3}[K_{i_3}+1] 
		&\leq (1-\alpha)\overline{x}[k_1]+\alpha (\overline{x}[k_1]-\alpha^{k_t-k_1} \overline{\epsilon}_0) \nonumber \\
		&\leq \overline{x}[k_1]-\alpha^{1+k_t-k_1} \overline{\epsilon}_0, \nonumber \\
		x_{i_3}[K_{i_3}+2] &\leq \overline{x}[k_1]-\alpha^{2+k_t-k_1} \overline{\epsilon}_0, \nonumber \\
		&  \hspace*{0.19cm}\mbox{} \vdots    \nonumber \\
		x_{i_3}[k_{t+1}] &\leq \overline{x}[k_1]-\alpha^{k_{t+1}-k_1} \overline{\epsilon}_0.
	\end{align}
	This implies that
	\begin{align}\label{x_notin}
		\mathcal{X}_t \cap \mathcal{U}(k_{t+1}, \overline{\epsilon}_{k_{t+1}-k_1},\underline{\epsilon}_{k_{t+1}-k_1}) = \emptyset.
	\end{align}
	Note that the bounds in \eqref{lower_pk} and \eqref{upper_pk} also apply to all nodes $j_3\in \overline{\mathcal{U}}(k_t,\overline{\epsilon}_{k_t-k_1},\underline{\epsilon}_{k_t-k_1})$, since they have values located in $[\underline{x}[k_1]+\alpha^{k_t-k_1} \underline{\epsilon}_0, \overline{x}[k_1]-\alpha^{k_t-k_1} \overline{\epsilon}_0]$ by definition, and each $j_3$ does not filter out its own value. Therefore, we have
	\begin{align}\label{ubar_notin}
		\overline{\mathcal{U}}(k_t,\overline{\epsilon}_{k_t-k_1},\underline{\epsilon}_{k_t-k_1})  \cap 
		\mathcal{U}(k_{t+1}, \overline{\epsilon}_{k_{t+1}-k_1},\underline{\epsilon}_{k_{t+1}-k_1}) = \emptyset.
	\end{align}
	Hence, \eqref{x_notin} and \eqref{ubar_notin} together have proved that \eqref{shrinking} holds for $t\geq 3$.

	Since $\mathcal{W}^\mathcal{N}$ is finite, there must be a time step $k_1 + (|\mathcal{W}^\mathcal{N}|+1) K$ such that $\mathcal{U}(k_1+(|\mathcal{W}^\mathcal{N}|+1)K, \overline{\epsilon}_{(|\mathcal{W}^\mathcal{N}|+1)K},\underline{\epsilon}_{(|\mathcal{W}^\mathcal{N}|+1)K}) = \emptyset$. This implies that $\forall i\in\mathcal{W}^\mathcal{N}$, we have
	\begin{align}\label{lower_pk11}
		x_{i}[k_1+(|\mathcal{W}^\mathcal{N}|+1)K] &\geq \underline{x}[k_1]+\alpha^{(|\mathcal{W}^\mathcal{N}|+1)K} \underline{\epsilon}_0, \nonumber \\
		x_{i}[k_1+(|\mathcal{W}^\mathcal{N}|+1)K] &\leq \overline{x}[k_1]-\alpha^{(|\mathcal{W}^\mathcal{N}|+1)K} \overline{\epsilon}_0.
	\end{align}
	It further indicates that
	\begin{align}\label{lower_pk22}
		V&[k_1+(|\mathcal{W}^\mathcal{N}|+1)K]  \nonumber \\
		&=\overline{x}_{i}[k_1+(|\mathcal{W}^\mathcal{N}|+1)K] - \underline{x}_{i}[k_1+(|\mathcal{W}^\mathcal{N}|+1)K] \nonumber \\
		&\leq \overline{x}[k_1]-\alpha^{(|\mathcal{W}^\mathcal{N}|+1)K} \overline{\epsilon}_0  - (\underline{x}[k_1]+\alpha^{(|\mathcal{W}^\mathcal{N}|+1)K} \underline{\epsilon}_0) \nonumber \\
		&= V[k_1]-\alpha^{(|\mathcal{W}^\mathcal{N}|+1)K} (\overline{\epsilon}_0+\underline{\epsilon}_0).
	\end{align}
	Recall that $\underline{\epsilon}_0=r[k_1] - \underline{x}[k_1]$ and $\overline{\epsilon}_0=\overline{x}[k_1] - r[k_1] $, and hence, $\overline{\epsilon}_0+\underline{\epsilon}_0= V[k_1]$. Then we have
	\begin{align}\label{lower_pk33}
		V[k_1+(|\mathcal{W}^\mathcal{N}|+1)K] &\leq V[k_1]-\alpha^{(|\mathcal{W}^\mathcal{N}|+1)K} V[k_1] \nonumber \\
		&= (1-\alpha^{(|\mathcal{W}^\mathcal{N}|+1)K}) V[k_1].
	\end{align}
	We can repeat the above analysis and obtain for $\delta \in \mathbb{Z}_{> 1}$,
	\begin{align}\label{lower_pk44}
		V&[k_1+(|\mathcal{W}^\mathcal{N}|+1)\delta K]  \nonumber \\
		&\leq (1-\alpha^{(|\mathcal{W}^\mathcal{N}|+1)K}) V[k_1+(|\mathcal{W}^\mathcal{N}|+1)(\delta-1) K] \nonumber \\
		&\leq (1-\alpha^{(|\mathcal{W}^\mathcal{N}|+1)K})^\delta  V[k_1] .
	\end{align}
	%where $k_1+(|\mathcal{W}^\mathcal{N}|+1)\delta K < k_e$.
	When time tends to infinity, by Lemma~\ref{lemma_safety}, we further have 
	\begin{align}\label{lower_pk55}
		\lim_{k\to \infty} V[k] &\leq \lim_{\delta \to \infty} V[k_1+(|\mathcal{W}^\mathcal{N}|+1)\delta K]  \nonumber \\
		&\leq \lim_{\delta \to \infty} (1-\alpha^{(|\mathcal{W}^\mathcal{N}|+1)K})^\delta  V[k_1] = 0,
	\end{align}
	where the last equality holds since $0<1-\alpha^{(|\mathcal{W}^\mathcal{N}|+1)K} < 1$. This clearly indicates that all normal nodes reach consensus as in \eqref{reach_consensus}. The proof is complete.
\end{proof}

\begin{remark}
		For the special case when there is no adversarial node in the network $\mathcal{G}[k]$ (i.e., $f=0$), our result in Theorem~\ref{theorem_firstorder} implies that followers using the MW-MSR algorithm can also achieve leader-follower consensus if and only if $\mathcal{G}[k]$ is a jointly $1$-robust following graph with $l$ hops. We note that for the fault-free case, this graph condition is equivalent to the one in \cite{ren2008consensus}, which is that $\mathcal{G}[k]$ jointly has a spanning tree rooted at the leader.\footnote{If there are multiple leaders in $\mathcal{G}[k]$, they can be viewed as one node.}  
\end{remark}

\begin{remark}
	We emphasize that our graph condition is tight for the MW-MSR algorithm applied in time-varying leader-follower networks. In fact, even for the one-hop case, it is tighter than the one in \cite{usevitch2020resilient} in terms of static networks as indicated in Lemma~\ref{tighter1}. Specifically, the authors of \cite{usevitch2020resilient} have looked into the resilient leader-follower consensus with a static reference value using the SW-MSR algorithm. They have accordingly proposed a sufficient graph condition for their algorithm to succeed, which is that $\mathcal{G}[k]$ is strongly $(Q, t_0, 2f+1)$-robust w.r.t. the set $\mathcal{L}$. It requires the union of $\mathcal{G}[k]$ across a time interval of length $Q$ to satisfy that any nonempty node subset $\mathcal{C} \subseteq \mathcal{V}\setminus \mathcal{L}$ contains at least one node having $2f+1$ in-neighbors outside $\mathcal{C}$.
\end{remark}

We can also apply the multi-hop relay technique in the SW-MSR algorithm, and our sufficient condition for the new algorithm can also be relaxed compared to the one in \cite{usevitch2020resilient}. The proof of Lemma~\ref{tighter1} is given in the Appendix A.

\begin{lemma}\label{tighter1}
	If graph $\mathcal{G}[k]$ is strongly $(Q, t_0, 2f+1)$-robust w.r.t. the set $\mathcal{L}$ where $Q=1$, then $\mathcal{G}[k]$ is a jointly $(f+1)$-robust following graph with $1$ hop where $k_{t+1}-k_t=1,  \forall [k_t, k_{t+1})$, and the converse does not hold.
\end{lemma}

We note that a numerical example is provided in Section~\ref{sec_example1} of the graph described in the proof. We apply the one-hop MW-MSR algorithm in this network, which is equivalent to the SW-MSR algorithm in \cite{usevitch2020resilient} when $Q=1$. The results show that resilient leader-follower consensus is achieved and this claim is verified from the numerical viewpoint.

\begin{remark}
Our result generalizes the one for the one-hop W-MSR algorithm applied in time-varying leader-follower networks. Besides, our graph condition is more relaxed for the case of $l\geq 2$ than that of the one-hop case since joint graph robustness generally increases as the relay range $l$ increases. We will illustrate this point through numerical examples in Section~\ref{sec_example}. 
Moreover, in this paper, we study the Byzantine model, which is more adversarial than the malicious model studied for leaderless resilient consensus in \cite{leblanc2013resilient,saldana2017resilient,yuan2021resilient}.
Similar to the works studying the Byzantine model \cite{yuan2022asynchronous,su2017reaching,vaidya2012iterative}, our graph condition guarantees that the normal network is sufficiently robust to fight against Byzantine agents for tracking the correct reference value.
\end{remark}

\begin{remark}\label{remark_leaderless}
		We compare the analysis in Theorem~\ref{theorem_firstorder} with the leaderless resilient consensus results. For convenience, we consider static networks in this comparison, and Corollary~\ref{corollary_static} in Section~\ref{sec_second} will formally state our results for such a case.
		In the works \cite{leblanc2013resilient,yuan2021resilient} studying leaderless resilient consensus under the malicious model, the key graph condition of $r$-robustness with $l$ hops requires that for any two nonempty disjoint node sets $\mathcal{V}_1, \mathcal{V}_2 \subset \mathcal{V}$, at least one set must contain a node having $r$ independent paths originating from nodes outside the set. For the Byzantine model, the condition of strict robustness requires a similar structure \cite{yuan2022asynchronous}. In contrast, the condition of robust following graphs is on a single arbitrary set $\mathcal{S} \subseteq \mathcal{H}\setminus \mathcal{L}$ in Definition~\ref{robust_following}. This difference is due to the nature of the two problems. In particular, leaderless resilient consensus requires normal nodes to reach consensus on a value which is not determined a priori. Thus, the normal nodes in $\mathcal{V}_1$ (or $\mathcal{V}_2$) either are influenced by or influence the ones outside the set. Therefore, the robustness notion is defined on two node sets to characterize such two possible information flows. This also indicates that some normal nodes in $\mathcal{V}_1$ (or $\mathcal{V}_2$) may not receive enough influences from the outside of the set and such nodes can act as ``virtual leaders'' in the network and resilient consensus can still be achieved.
		On the other hand, resilient leader-follower consensus requires normal followers to reach consensus on the leaders' value. Hence, the followers in every set $\mathcal{S}$ must be influenced by the normal nodes outside the set through enough incoming paths.
\end{remark}

Stemming from this important difference, we can conclude that also for the malicious model, the graph condition in Theorem~\ref{theorem_firstorder} is necessary and sufficient for the MW-MSR algorithm to guarantee resilient leader-follower consensus. This is because for the leader-follower case, the types of adversaries will not change the graph structure for normal nodes in set $\mathcal{S}$ to be influenced by the normal ones outside $\mathcal{S}$.

%\begin{remark}
%It is worth noting that leader-follower consensus requires each follower node to track the reference value. Hence, each follower $i\in \mathcal{S}$ must have enough incoming paths (or edges) from the outside of the node set $\mathcal{S}$ it belongs to. This is very different from leaderless resilient consensus under the malicious model, where some normal nodes in the network may not receive enough influences from the outside of the node set they belong to and such nodes can act as ``virtual leaders'' in the network and further achieve resilient consensus. 
%\textcolor{blue}{ 
%Therefore, we conclude that the graph condition in Theorem~\ref{theorem_firstorder} is also necessary and sufficient for the MW-MSR algorithm to guarantee resilient leader-follower consensus under the malicious model.}
%\end{remark}

\subsection{Analysis on Secure Leader Agents and Comparisons with Related Works}

%So far, our discussions have focused on the cases where the leader agents are not secure, i.e., it is possible that $|\mathcal{L}^\mathcal{N}| < |\mathcal{L}|$. There are literature studying cases with secure leader agent(s) (i.e., $\mathcal{L}^\mathcal{N} = \mathcal{L}$), where such a leader is guaranteed to be normal through additional security measures so that it is not vulnerable to failures or attacks. For example, the recent work \cite{rezaee2021resiliency} studied resilient dynamic leader-follower consensus where the secure leader agent broadcasts a time-varying reference value and the authors provided a sufficient graph condition for their algorithm. Moreover, the authors of \cite{tseng2015broadcast} proved a necessary and sufficient graph condition for the CPA algorithm to broadcast a static reference value of the secure leader to followers in the network.

In this subsection, we turn our attention to the secure leader case with $\mathcal{L}^\mathcal{N} = \mathcal{L}$. In the literature, this case has been studied. Agents can be secure through additional security measures so that they are not vulnerable to failures or attacks. 
For example, the work \cite{rezaee2021resiliency} studied resilient dynamic leader-follower consensus where the secure leader agent broadcasts a time-varying reference value and provided a sufficient graph condition for followers to track the leader with bounded errors. Moreover, the authors of \cite{tseng2015broadcast} proved a necessary and sufficient graph condition for the CPA algorithm where a static reference value of the secure leader is broadcast to followers in the network.

%We formulate our problem of resilient leader-follower consensus with secure leader agents (i.e., Case 1 in Table~I) as follows. For Case 2 there, since secure leaders are not known by the followers, the followers have to apply the MW-MSR algorithm as well in order to avoid being misled by the adversarial neighbors. Thus, the approach to handle Case 2 is the same as the one for Cases 3 and 4 in Table~I.

We formulate our resilient leader-follower consensus with secure leader agents as Problem~\ref{problem2} below. Here, we study the special case when such leaders are known to the followers. This is indeed the simplest problem setting to handle among the four cases listed in Table~I. The other case is when the secure leaders are unknown. 
We however note that this case has to be treated as in the cases of insecure leaders since the followers have to apply the MW-MSR algorithm as well in order to avoid being misled by the adversarial neighbors. Thus, the approach to handle Case~2 is the same as the one for Cases~3 and 4 in Table~I.

\begin{problem}\label{problem2}
	Suppose that all leader agents in $\mathcal{L}$ are secure and their identities are known to the direct followers, which are out-neighbors of any leader agent.
	We say that the normal agents in $\mathcal{N}$ reach resilient leader-follower consensus with secure leader agents if for any possible sets and behaviors of the adversaries in $\mathcal{A}$ and any state values of the normal agents in $\mathcal{N}$, \eqref{reach_consensus} is satisfied.
\end{problem}

Denote by $\mathcal{U}_{\mathcal{G}}[k_t, k_{t+1})$ the union of the graphs $\mathcal{G}[k]$ across the time interval $[k_t, k_{t+1})$. Moreover, let $\mathcal{W}_\mathcal{L} =  \{i\in \mathcal{W} :  \exists d \in \mathcal{L} \medspace \textup{s.t.} \medspace  i\in \mathcal{N}_d^{1+} \} $ in the union graph $\mathcal{U}_{\mathcal{G}}[0, \infty)$.
To solve the problem with secure leader agents, we slightly modify our resilient algorithm as follows.

\textit{The modified MW-MSR algorithm:} Each normal follower $i\in \mathcal{W}_\mathcal{L} \cap \mathcal{N}$ updates its value $x_i[k+1]=r[k]$ at each time $k$. Meanwhile, the rest of the normal followers in $\mathcal{W}^\mathcal{N}\setminus \mathcal{W}_\mathcal{L}$ still apply the MW-MSR algorithm.

From the modified algorithm and the assumption on the secure leader agents, we can see that each follower in $\mathcal{W}_\mathcal{L}$ need not have other connections except at least one incoming edge from the secure leader agents. Therefore, they can be viewed as the ``not secure leaders'' for the rest of the followers. Hence, we denote the reduced subgraph of $\mathcal{G}[k]$ by $\mathcal{G}_\mathcal{S}[k]$, where $\mathcal{S}=\mathcal{V}\setminus\mathcal{L}$. Moreover, denote the set of such virtual leaders in $\mathcal{G}_\mathcal{S}[k]$ by $\mathcal{L}'=\mathcal{W}_\mathcal{L}$.

We are ready to state a result for the modified MW-MSR algorithm to achieve resilient leader-follower consensus with secure leader agents. Since it can be proved using an analysis similar to that in the proof of Theorem~\ref{theorem_firstorder}, we omit its proof.  
Moreover, we note that the results for secure leader agents in Proposition~\ref{proposition_secure} also hold for the MDP-MSR algorithm to be introduced in Section~\ref{sec_second} studying second-order MASs.

\begin{proposition}\label{proposition_secure}
	Consider the time-varying network $\mathcal{G}[k] = (\mathcal{V},\mathcal{E}[k])$ with $l$-hop communication, where each normal follower node $i\in \mathcal{W}^\mathcal{N}$ updates its value according to the modified MW-MSR algorithm with parameter $f$. Under the $f$-local adversarial set $\mathcal{A}$ and the assumption that $r [k]$ is constant $\forall k\in [k_1, \infty)$, resilient leader-follower consensus with secure leader agents is achieved if and only if the subgraph $\mathcal{G}_\mathcal{S}[k]$ is a jointly $(f+1)$-robust following graph with $l$ hops.
\end{proposition}

\begin{remark}\label{tighter35}
	We emphasize that our graph condition for static networks (i.e., the subgraph $\mathcal{G}_\mathcal{S}$ is an $(f+1)$-robust following graph with $1$ hop) in Proposition~\ref{proposition_secure} is tighter than the sufficient condition in \cite{rezaee2021resiliency} studying resilient dynamic leader-follower consensus in static networks. We first claim that the condition for static networks in \cite{usevitch2020resilient} is equivalent to the one in \cite{rezaee2021resiliency} with the secure leader removed. We can easily see this from definitions of the two conditions. The condition in \cite{rezaee2021resiliency} is that $\mathcal{G}$ is a $(2f +1)$-robust leader-follower graph, which requires that $|\mathcal{W}_\mathcal{L}|\geq 2f+1$ and any nonempty set $\mathcal{S} \subseteq \mathcal{W} \setminus \mathcal{W}_\mathcal{L}$ contains at least one node having $2f+1$ in-neighbors outside $\mathcal{S}$. Meanwhile, the condition for static networks in \cite{usevitch2020resilient} is that $\mathcal{G}[k]$ is strongly $(Q, t_0, 2f+1)$-robust w.r.t. the set $\mathcal{L}$ where $Q=1$; this requires that $|\mathcal{L}|\geq 2f+1$ and any nonempty set $\mathcal{S} \subseteq \mathcal{V} \setminus \mathcal{L}$ contains at least one node having $2f+1$ in-neighbors outside $\mathcal{S}$. Invoking Lemma~\ref{tighter1}, we conclude that our graph condition for static networks in Proposition~\ref{proposition_secure} is tighter than the one in \cite{rezaee2021resiliency}.
\end{remark}

From Remark~\ref{tighter35}, one can see that the leader-follower graph structure for static references \cite{usevitch2020resilient} and that for dynamic references \cite{rezaee2021resiliency} are the same. The reason is that the graph structure characterizes the information flow of leader-follower consensus. Meanwhile, agents' dynamics may vary from tracking a static reference to tracking a dynamic reference.

\begin{remark}
	We compare our results for secure leader agents with several related works.
	Our graph condition for static networks in Proposition~\ref{proposition_secure} is equivalent to the necessary and sufficient condition for the CPA algorithm in \cite{tseng2015broadcast} to succeed. Moreover, our necessary and sufficient condition with $1$ hop for static networks is tighter than the sufficient conditions in \cite{usevitch2020resilient,rezaee2021resiliency} for static networks. 
	Besides, for the multi-hop case, our condition with $l\geq 2$ hops is even tighter than the conditions in the literature \cite{tseng2015broadcast,usevitch2020resilient,rezaee2021resiliency} since the graph robustness generally increases (and definitely does not decrease) as the relay range $l$ increases.
\end{remark}

\subsection{Properties of Jointly (f+1)-Robust Following Graphs}

In the following lemma, we list several properties of jointly $(f+1)$-robust following graphs with $l$ hops. Its proof can be found in the Appendix B.

\begin{lemma}\label{2f+1leaders}
	If graph $\mathcal{G}[k]$ is a jointly $(f+1)$-robust following graph with $l$ hops under the $f$-local model, then the following hold:
	\begin{enumerate}
		\item $|\mathcal{L}| \geq 2f+1$.
		\item $\forall [k_t, k_{t+1}), \exists K_i \in [k_t, k_{t+1})$, $\exists i \in \mathcal{W}^\mathcal{N}$ s.t. $|\mathcal{N}_i^{l-}[K_i] \cap \mathcal{L}| \geq 2f+1$.
		\item $|\mathcal{W}_\mathcal{L}[k_t, k_{t+1})| = | \{i\in \mathcal{W} :  \exists d \in \mathcal{L} \medspace \textup{s.t.} \medspace  i\in \mathcal{N}_d^{1+}[k_t, k_{t+1}) \} | \geq 2f+1$ in graph $\mathcal{U}_{\mathcal{G}}[k_t, k_{t+1})$.
		\item $\forall [k_t, k_{t+1}), \exists K_i \in [k_t, k_{t+1})$ s.t. $|\mathcal{N}_i^{1-}[K_i]| \geq 2f+1, \forall i \in \mathcal{W}$. Moreover, the minimum number of directed edges of $\mathcal{U}_{\mathcal{G}}[k_t, k_{t+1})$ with minimum $|\mathcal{L}|$ is $(2f+1)|\mathcal{W}|$.
	\end{enumerate}
\end{lemma}
\vspace{2mm}

From Lemma~\ref{2f+1leaders}, we see that there are several necessary graph conditions for our algorithms to achieve resilient leader-follower consensus under the $f$-local adversarial model. First, there must be at least $2f+1$ leaders in $\mathcal{L}$. Second, there must be at least one follower node having $2f+1$ leaders in its $l$-hop neighbors at some time step in each $[k_t, k_{t+1})$. Third, there must be at least $2f+1$ follower nodes having incoming edges directly from the leaders in each $[k_t, k_{t+1})$. Lastly, each follower node must have $2f+1$ incoming edges at some time step in each $ [k_t, k_{t+1})$.

Moreover, in the special case of the static network $\mathcal{G}$, for our algorithms to achieve leader-follower consensus, these conditions must hold at each time instant. In particular, it must have directed edges no less than $(2f+1)|\mathcal{W}|$. This requirement is consistent with the one reported in \cite{rezaee2021resiliency}. However, we emphasize that our graph condition is tighter as mentioned earlier and covers a wider range of graphs satisfying our condition. Furthermore, we can utilize undirected edges to relax the heavy connectivity requirement. For example, in Fig.~\ref{15node}, $\mathcal{G}[k]$ has 28 directed/undirected edges while it needs 50 directed edges to satisfy the condition in \cite{rezaee2021resiliency} for followers.

\section{Resilient Leader-Follower Consensus in Second-Order MASs}\label{sec_second}

In this section, we propose a novel algorithm to deal with agents having second-order dynamics and to solve resilient leader-follower consensus in time-varying networks. We characterize tight graph conditions for the algorithm to succeed under the Byzantine model. 

We consider a second-order MAS with time-varying communication network $\mathcal{G}[k] = (\mathcal{V},\mathcal{E}[k])$. 
Each follower node $i\in \mathcal{W}$ has a double-integrator dynamics from \cite{dibaji2017resilient,ren2011distributed}, whose discretized form is given as 
\begin{align}\label{secondorder}
	\hat{x}_i[k+1]&= \hat{x}_i[k] + Tv_i[k] + \frac{T^2}{2} u_i[k], \nonumber \\ 
	v_i[k+1]&= v_i[k] + Tu_i[k],
\end{align}
where $v_i[k]$, $u_i[k] \in \mathbb{R}$, and $T$ are, respectively, the velocity, the control input of node $i$ at time $k$, and the sampling period. Moreover, $\hat{x}_i[k]=x_i[k]- \delta_i$, where $x_i[k] \in \mathbb{R}$ is the absolute position of node $i$ and $\delta_i \in \mathbb{R}$ is a constant representing the desired relative position of node $i$ in a formation. For the sake of simplicity, we call $\hat{x}_i[k]$ as the agents’ positions.

Meanwhile, leaders in $\mathcal{L}$ have the same dynamics as \eqref{secondorder}. However, since we consider the static reference value in this paper, we assume that each normal leader agent $d\in \mathcal{L}^\mathcal{N}$ updates its position according to the reference function in \eqref{leader} (i.e., $v_d[k]=0, u_d[k]=0, \forall k\in [k_1, \infty)$) and propagates $\hat{x}_d[k]=x_d[k]- \delta_d$ to followers, where $\delta_d \in \mathbb{R}$ is the desired relative position of leader $d$ in a formation.

At each time $k$, the control input $u_i[k]$ of each follower node $i \in\mathcal{W}$ utilizes the relative positions of its neighbors and its own velocity \cite{ren2011distributed}:
\begin{align}\label{normalupdate2}
	u_i[k]=\sum_{j =1}^{n} a_{ij}[k] (\hat{x}_j[k] - \hat{x}_i[k]) -\beta v_i[k],  
\end{align}
where $\beta$ is a positive constant, and $a_{ij}[k] > 0$ if $j \in \mathcal{N}_i^{l-}[k]$ and $a_{ij}[k] = 0$ otherwise.

When there is no attack, we can derive from the results in \cite{ren2011distributed} that if in graph $\mathcal{G}[k]$, all followers jointly have directed paths originating from at least one leader across each time interval $[k_t, k_{t+1})$ and with $\beta$ and $T$ properly chosen, then all followers reach consensus on the reference position value in the sense that they come to formation and stop asymptotically:
\begin{align*}
	x_i[k]- x_d[k] &\to \delta_i -\delta_d, \\
	v_i[k] &\to 0 \medspace\medspace \textup{as}  \medspace\medspace k \to \infty, \medspace \forall i \in \mathcal{W}, \forall d \in \mathcal{L}.
\end{align*}

\begin{algorithm}[t]
	\caption{MDP-MSR Algorithm }
	\LinesNumbered 
	\KwIn{Node $i$ knows $\hat{x}_i[0]$, $\mathcal{N}_i^{l-}[k]$, $\mathcal{N}_i^{l+}[k]$. }
	
	\For{$k\geq0$}{
		
		\SetKwBlock{newbox}{1) Exchange messages:}{}
		\newbox{
			\SetAlgoVlined
			Send $m_{ij}[k]=(\hat{x}_i[k],P_{ij}[k])$ to $\forall j\in \mathcal{N}_i^{l+}[k]$. 
			
			Receive $m_{ji}[k]=(\hat{x}_j[k],P_{ji}[k])$ from $\forall j\in \mathcal{N}_i^{l-}[k]$ and store them in $\mathcal{M}_i[k]$.
			
			Sort $\mathcal{M}_i[k]$ in an increasing order based on the message values (i.e., $\hat{x}_j[k]$ in $m_{ji}[k]$).
		}

		\textbf{2) Remove extreme values using Step 2 of Algorithm~1}
		
		\SetKwBlock{newbox}{3) Update values using \eqref{secondorder} with control input $u_i[k]$ given by:}{}
		\newbox{
			\SetAlgoVlined
			$\medspace\medspace a_{i}[k]=1/(\left| \mathcal{M}_i[k]\setminus \mathcal{R}_i[k] \right| )$,
			\begin{align}\label{msrupdate2}
				u_i[k]&=\sum_{m\in \mathcal{M}_i[k]\setminus \mathcal{R}_i[k]} a_{i}[k] (\mathrm{value}(m) - \hat{x}_i[k])   \nonumber\\
				& \medspace\medspace \medspace\medspace\medspace     -\beta v_i[k].  
			\end{align}
		}
		\KwOut{$\hat{x}_i[k+1]$, $v_i[k+1]$.}
	}
	%\vspace*{-1.0mm}
\end{algorithm}

In our resilient leader-follower consensus problem of the second-order MAS \eqref{secondorder}, adversary nodes may not follow the update rule \eqref{normalupdate2} and even sends faulty values to neighbors to prevent the normal nodes from reaching consensus. The problem is the same as Problem~\ref{problem} except that agents exchange $\hat{x}$ values with neighbors.

For this problem, we present a novel algorithm called the Multi-hop Double-integrator Position-based MSR (MDP-MSR) algorithm in Algorithm~2. At each time $k$, each follower $i \in \mathcal{W}^\mathcal{N}$ exchanges $\hat{x}_i[k]$ with its neighbors within $l$ hops and utilizes the MMC technique to filter away the extreme values as in Algorithm~1. Finally, it updates its value using the remaining ones.

The control input in \eqref{msrupdate2} can be transformed to the following form:
\begin{align}\label{msrupdate22}
	u_i[k]=\sum_{j =1}^{n_N} a_{ij}[k] (\hat{x}_j[k] - \hat{x}_i[k]) -\beta v_i[k].  
\end{align}
where $a_{ij}[k] > 0$ if $m_{ji}[k]\in \mathcal{M}_i[k]\setminus \mathcal{R}_i[k]$ and $a_{ij}[k] = 0$ otherwise. Moreover, for $T$ and $\beta$, we assume that
\begin{align}\label{BT}
	1+ \frac{T^2}{2}  \leq \beta T \leq  2- \frac{T^2}{2}.
\end{align}
Regarding the leader-follower consensus error, define three variables $\forall k\geq 1$:
\begin{align}\label{consensus_error2}
	\overline{z}[k]&= \max_{i\in \mathcal{W}^\mathcal{N}, \medspace d\in \mathcal{L}^\mathcal{N}} \medspace \{ \hat{x}_d[k],\hat{x}_d[k-1], \hat{x}_i[k] , \hat{x}_i[k-1] \}, \nonumber \\[1mm]
	\underline{z}[k]&= \min_{i\in \mathcal{W}^\mathcal{N}, \medspace d\in \mathcal{L}^\mathcal{N}} \medspace \{ \hat{x}_d[k],\hat{x}_d[k-1], \hat{x}_i[k] , \hat{x}_i[k-1] \}, \nonumber \\[1mm]
	\hat{V}[k] &= \overline{z}[k]-\underline{z}[k]. 
\end{align}

In the next theorem, we provide a necessary and sufficient graph condition for the MDP-MSR algorithm to achieve resilient leader-follower consensus in time-varying networks.

\begin{theorem}\label{theorem_secondorder}
	Consider the time-varying network $\mathcal{G}[k] = (\mathcal{V},\mathcal{E}[k])$ with $l$-hop communication, where each normal follower node $i\in \mathcal{W}^\mathcal{N}$ updates its value according to the MDP-MSR algorithm with parameter $f$. Under the $f$-local adversarial set $\mathcal{A}$ and the assumption that $r [k]$ is constant $\forall k\in [k_1, \infty)$, resilient leader-follower consensus is achieved if and only if $\mathcal{G}[k]$ is a jointly $(f+1)$-robust following graph with $l$ hops.
\end{theorem}

\begin{proof}
	\textit{(Necessity)} If $\mathcal{G}[k]$ is not a jointly $(f+1)$-robust following graph with $l$ hops, by the same reasoning as in the necessity proof of Theorem~\ref{theorem_firstorder}, we obtain that there is a nonempty subset $\mathcal{S}\subseteq \mathcal{N}\setminus \mathcal{L}$ such that \eqref{atmostf} holds.

	Suppose that $\hat{x}_i[\hat{k}]=a , \medspace \forall i\in \mathcal{S}$, and $\hat{x}_j[\hat{k}]=r[k_1], \medspace \forall j\in \mathcal{N}\setminus \mathcal{S}$, where $a< r[k_1]$ is constant. Moreover, $v_i[\hat{k}]=0, \medspace \forall i\in \mathcal{V}$. Assume that the Byzantine nodes send $a$ and $r[k_1]$ to the nodes in $\mathcal{S}$ and $ \mathcal{N}\setminus \mathcal{S}$, respectively. 
	For any normal node $i \in \mathcal{S}$, it removes all the values of neighbors outside $\mathcal{S}$ since the message cover of these values has cardinality equal to $f$ or less. According to Algorithm~2, such normal nodes will keep their values, i.e., $\hat{x}_i[k]=a, v_i[k]=0, \medspace \forall i\in \mathcal{S}, \medspace \forall k\geq \hat{k}$. Thus, resilient leader-follower consensus is not achieved.

	\textit{(Sufficiency)}
	We first prove that $\big[ \underline{z}[k+1], \overline{z}[k+1] \big] \subset \big[ \underline{z}[k], \overline{z}[k] \big], \forall k\in [k_1, \infty)$. For $k \geq 1$, plugging \eqref{msrupdate22} into \eqref{secondorder}, we obtain $\forall i\in \mathcal{W}^\mathcal{N}$,
	\begin{align}
		\hat{x}_i[k+1]&= \hat{x}_i[k] + Tv_i[k] - \frac{T^2}{2} \beta v_i[k]    \nonumber \\ 
		&\medspace\medspace\medspace\medspace + \frac{T^2}{2} \sum_{j =1}^{n_N} a_{ij}[k] (\hat{x}_j[k] - \hat{x}_i[k]) ,  \label{msrupdatex} \\ 
		v_i[k+1]&= v_i[k] - T \beta v_i[k] + T\sum_{j =1}^{n_N} a_{ij}[k] (\hat{x}_j[k] - \hat{x}_i[k])  . \label{msrupdatev}
	\end{align}
	From \eqref{msrupdatex}, we notice that
	\begin{align}\label{msrupdatex1}
		\hat{x}_i&[k+1]- (1- T \beta)\hat{x}_i[k] \nonumber \\
		&=\hat{x}_i[k]- (1- T \beta)\hat{x}_i[k-1] \nonumber \\ 
		&\medspace\medspace\medspace\medspace + \frac{T^2}{2}\Big(\sum_{j =1}^{n_N} a_{ij}[k] (\hat{x}_j[k] - \hat{x}_i[k]) - (1- T \beta) \nonumber \\ 
		&\medspace\medspace\medspace\medspace  \times \sum_{j =1}^{n_N} a_{ij}[k-1] (\hat{x}_j[k-1] - \hat{x}_i[k-1]) \Big) \nonumber \\ 
		&\medspace\medspace\medspace\medspace + (T-\frac{T^2}{2} \beta) \big(v_i[k]- (1- T \beta)v_i[k-1] \big). 
	\end{align}
	Besides, from \eqref{msrupdatev}, we can obtain
	\begin{align}\label{msrupdatev1}
		v_i[k] &= (1- T \beta)v_i[k-1] \nonumber \\ 
		&\medspace\medspace\medspace\medspace +T\sum_{j =1}^{n_N} a_{ij}[k-1] (\hat{x}_j[k-1] - \hat{x}_i[k-1]) .
	\end{align}
	Thus, plugging \eqref{msrupdatev1} into \eqref{msrupdatex1}, we further have
	\begin{align}\label{msrupdatex2}
		\hat{x}_i[k+1] &=(2- T \beta)\hat{x}_i[k] + \frac{T^2}{2}\sum_{j =1}^{n_N} a_{ij}[k] (\hat{x}_j[k] - \hat{x}_i[k]) \nonumber \\ 
		&\medspace\medspace\medspace\medspace + \frac{T^2}{2} \sum_{j =1}^{n_N} a_{ij}[k-1] (\hat{x}_j[k-1] - \hat{x}_i[k-1]) \nonumber \\ 
		&\medspace\medspace\medspace\medspace      -(1- T \beta)\hat{x}_i[k-1] .
	\end{align}
	Combining \eqref{BT} and \eqref{msrupdatex2}, we conclude that $\forall i\in \mathcal{W}^\mathcal{N}$, $\hat{x}_i[k+1]$ is located in the convex combination of the positions of its normal neighbors within $l$ hops and itself at time $k$ and $k-1$. Notice that $\hat{x}_i[k] \in  \big[ \underline{z}[k], \overline{z}[k] \big], \hat{x}_i[k-1] \in  \big[ \underline{z}[k], \overline{z}[k] \big], \forall i\in \mathcal{W}^\mathcal{N}$ and $\hat{x}_d[k] = \hat{x}_d[k - 1] = r [k] , \forall d\in \mathcal{L}^\mathcal{N}, \forall k\in [k_1, \infty)$, which is also in $\big[ \underline{z}[k], \overline{z}[k] \big]$. Thus, it holds that $\hat{x}_i[k+1] \in  \big[ \underline{z}[k], \overline{z}[k] \big]$, i.e.,
	\begin{align*}
		\big[ \underline{z}[k+1], \overline{z}[k+1] \big] \subset \big[ \underline{z}[k], \overline{z}[k] \big], \forall k\in [k_1, \infty).
	\end{align*}

	Next, to further prove the convergence, follow an analysis as in the proof of Theorem~\ref{theorem_firstorder}, where we replace $x_i[k]$ with $\hat{x}_i[k]$, and replace $\overline{x}[k]$ and $\underline{x}[k]$ with $\overline{z}[k]$ and $\underline{z}[k]$, respectively. Moreover, define a lower bound $\omega$ such that  all the coefficients of $\hat{x}_i[k]$ and $\hat{x}_i[k-1], \forall i\in \mathcal{N}$ in \eqref{msrupdatex2} are larger than $\omega$ and further replace $\alpha$ with $\omega$ in the proof of Theorem~\ref{theorem_firstorder}.
	
	Thereafter, we can similarly prove that $\lim_{k\to \infty} \hat{V}[k] = 0$, and thus normal followers reach leader-follower consensus on their position values. Furthermore, when the normal followers reach leader-follower consensus on position values, from \eqref{msrupdatex}, we obtain $\hat{x}_i[k+1]\to \hat{x}_i[k]-(T-\beta T^2/2)v_i[k]$ as $k\to \infty$, $\forall i\in \mathcal{W}^\mathcal{N}$.
	Under \eqref{BT}, it holds that $v_i[k]\to 0$ as $k\to \infty$, $\forall i\in \mathcal{W}^\mathcal{N}$. Therefore, resilient leader-follower consensus is achieved.
\end{proof}

To deal with the special case of static networks, it is not hard to obtain the next corollary for the MW-MSR and MDP-MSR algorithms achieving resilient leader-follower consensus.

\begin{corollary}\label{corollary_static}
	Consider the static network $\mathcal{G} = (\mathcal{V},\mathcal{E})$ with $l$-hop communication, where all nodes $\forall i\in \mathcal{W}^\mathcal{N}$ update their values according to the MW-MSR or MDP-MSR algorithm with parameter $f$. Under the $f$-local adversarial set $\mathcal{A}$ and the assumption that $r [k]$ is constant $\forall k\in [k_1, \infty)$, resilient leader-follower consensus is achieved if and only if $\mathcal{G}$ is an $(f+1)$-robust following graph with $l$ hops.
\end{corollary}

	We see that the graph condition for the MDP-MSR algorithm under the Byzantine model is the same as the one for the MW-MSR algorithm.
	This feature shows the applicability of the key multi-hop techniques used in our algorithms. To the best of our knowledge, this paper is the first work studying resilient leader-follower consensus in MASs with second-order dynamics on agents. Besides, since our condition is defined on the original topology of the MAS, it can be verified prior to the deployment of our algorithm.
	From Theorem~\ref{theorem_secondorder} and Corollary~\ref{corollary_static}, one can see that the graph condition for time-varying networks is more relaxed than the one for static networks. Moreover, by introducing multi-hop relays to the static or time-varying MASs, we can relax the graph connectivity requirement for the one-hop MSR algorithms \cite{leblanc2013resilient,usevitch2020resilient,rezaee2021resiliency}. Through our approaches, the multi-hop relaying shows promising potentials for developing other types of resilient algorithms for enhancing security of MASs.

\begin{figure}[t]
	\centering
	%\vspace{-10pt}
	\subfigure[\scriptsize{The one-hop MW-MSR algorithm.}]{
		\includegraphics[width=3.3in,height=1.9in]{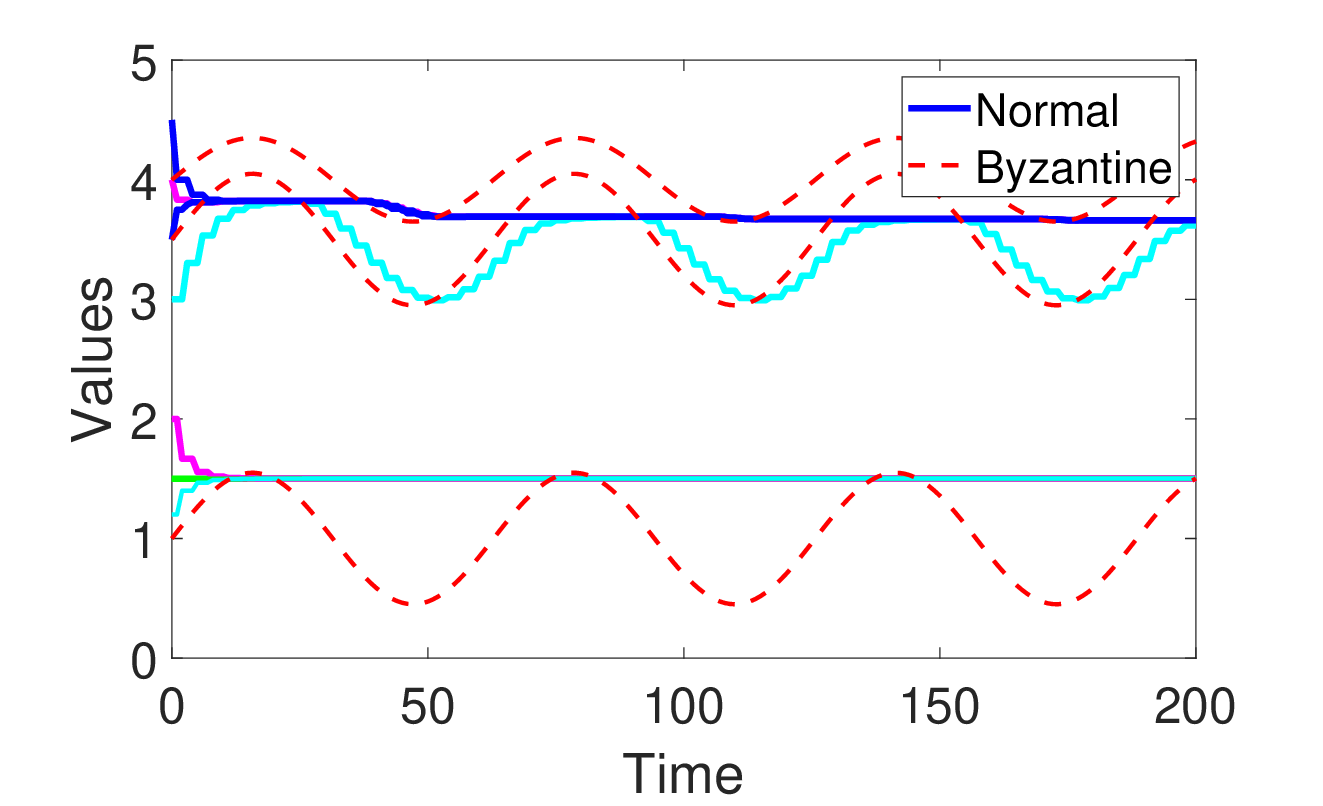}
	}
	\vspace{-3pt}
	
	\subfigure[\scriptsize{The three-hop MW-MSR algorithm.}]{
		\includegraphics[width=3.3in,height=1.9in]{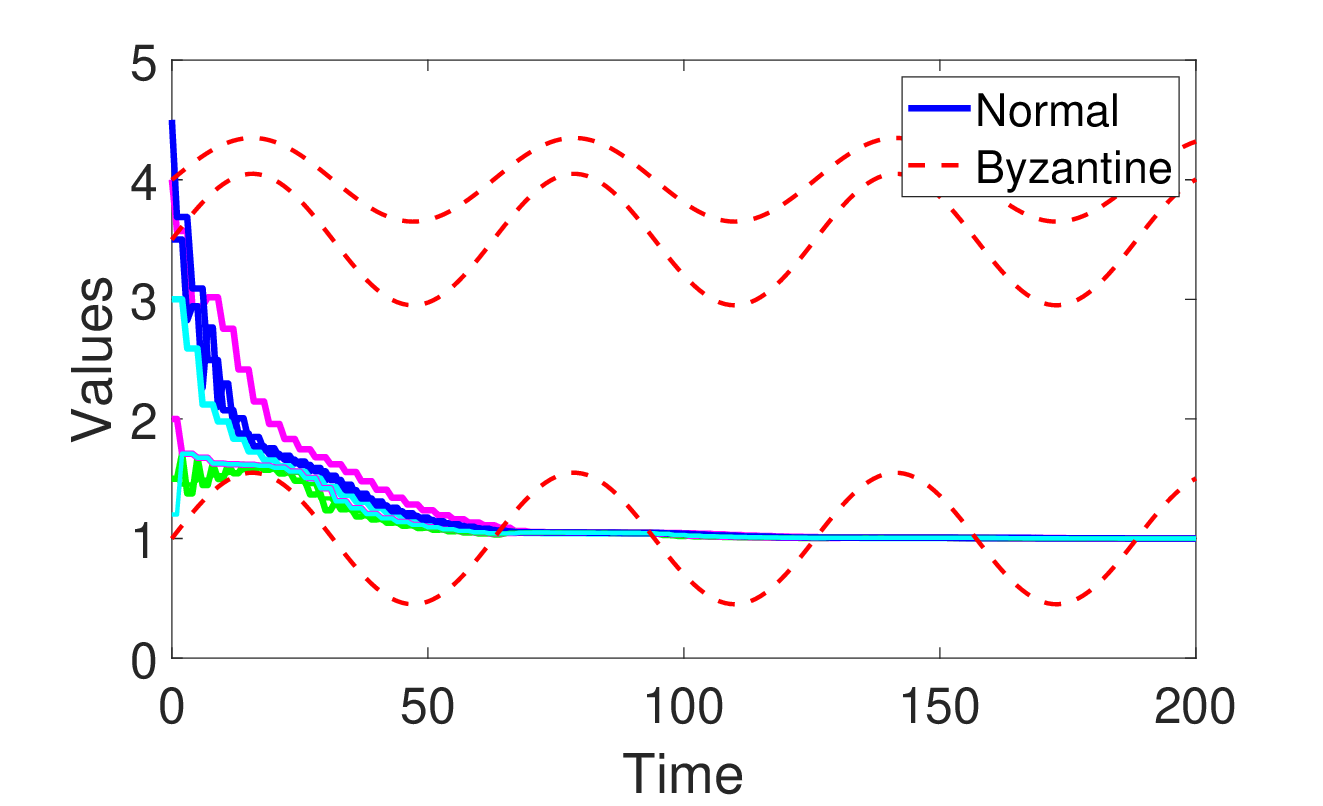}
	}
	\vspace{-3pt}
	\caption{Nodes' values in the time-varying leader-follower network in Fig.~\ref{15node} applying the MW-MSR algorithm.}
	\label{15nodevalue}
	\vspace*{-2.5mm}
\end{figure}

\section{Numerical Examples}\label{sec_example}

In this section, we conduct simulations for the MW-MSR algorithm and the MDP-MSR algorithm applied in time-varying leader-follower networks.

%\begin{figure}[t]
%	\centering
%	%\vspace{-10pt}
%	\includegraphics[width=3.5in,height=2in]{9_state_1hop_met}
%	
%	\vspace{-3pt}
%	\caption{Nodes' values in the time-varying network in Fig.~\ref{9node1} applying the one-hop MW-MSR algorithm.}
%	\label{9node_onehop}
%	\vspace*{-2.5mm}
%\end{figure}

\subsection{Simulation in A Time-Varying First-Order MAS}\label{sec_example1}

	In this part, we apply the MW-MSR algorithm to the time-varying leader-follower network in Fig.~\ref{15node} under the $2$-local model, i.e., with $f=2$.
	For the normal nodes, let $x_i[0]\in (1,5), \forall i \in \mathcal{W}^\mathcal{N}$ and $r[k]=1, \forall d \in \mathcal{L}^\mathcal{N}, \forall k\geq 0$. 
	This network is not a jointly $3$-robust following graph with $1$ hop. Hence, it is not robust enough for the one-hop W-MSR algorithm \cite{leblanc2013resilient,ishii2022overview,wen2023joint} to succeed under the $2$-local model. However, as we discussed in Example~\ref{discussion15node}, this network is a jointly $3$-robust following graph with $3$ hops.
	Here, we set nodes 7 and 8 to be Byzantine, and they send different values to different neighbors. Specifically, node 7 sends an oscillating value around the value of 4 to its neighbors. Moreover, node 8 sends an oscillating value around the value of 3.5 to out-neighbors in the node set $\{1, 2, 3\}$ and an oscillating value around 1 to out-neighbors in $\{4, 5\}$.

	We first apply the one-hop MW-MSR algorithm, which is equivalent to the W-MSR algorithm from \cite{leblanc2013resilient,wen2023joint}. 
	The results are given in Fig.~\ref{15nodevalue}(a), and resilient leader-follower consensus is not achieved. Note that in Fig.~\ref{15nodevalue} as well as all other figures in this section, lines in colors other than red represent the values (or trajectories) of normal followers.

\begin{figure}[t]
	\centering
	%\vspace{-10pt}
	\includegraphics[width=3.3in,height=1.9in]{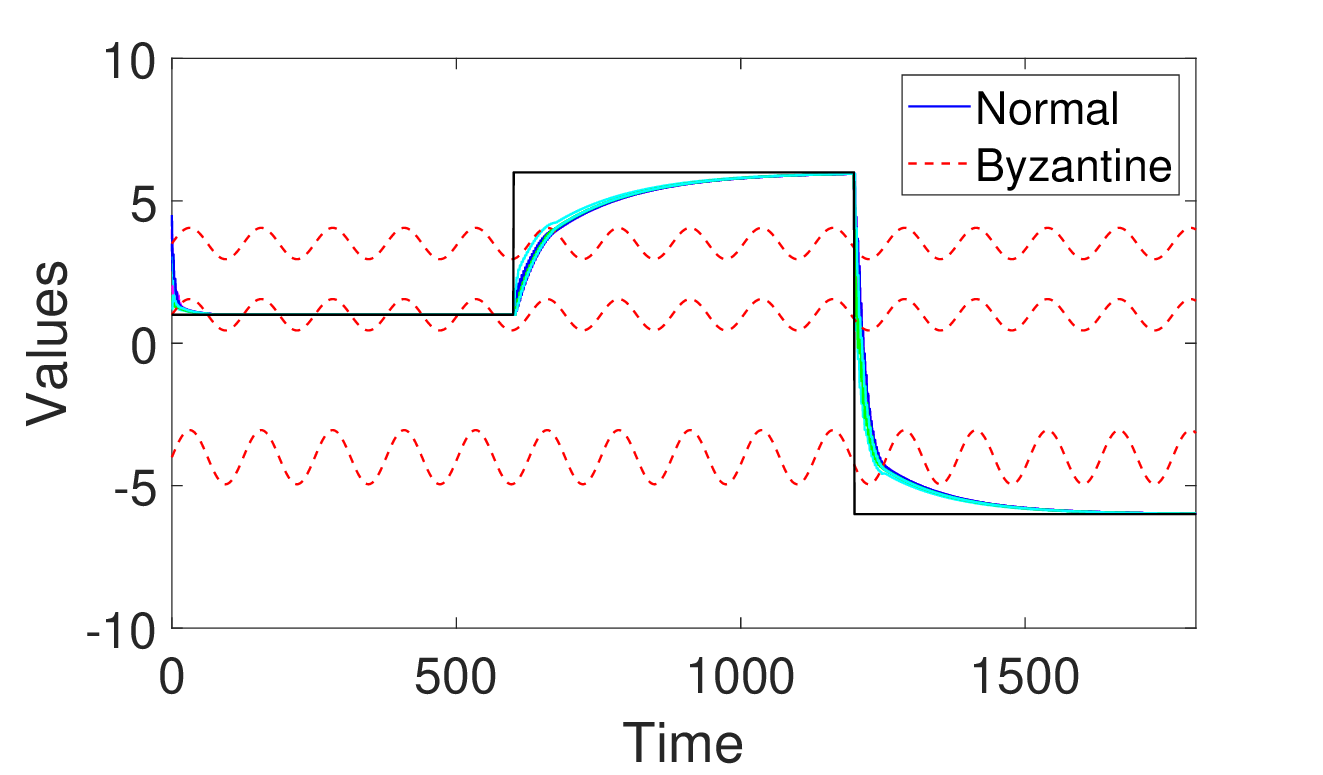}
	\vspace{-3pt}
	\caption{Normal nodes track the time-varying reference value in the time-varying leader-follower network in Fig.~\ref{15node} using the three-hop MW-MSR algorithm.}
	\label{15nodevalue_longtime}
	\vspace*{-2.5mm}
\end{figure}

Next, we apply the three-hop MW-MSR algorithm to this network. In this case, we assume that Byzantine nodes not only manipulate their own values as before but also relay false information. Specifically, when they receive a value from a neighbor and relay it to another neighbor, they manipulate it in the same way as they do to their own values.
We observe in Fig.~\ref{15nodevalue}(b) that resilient leader-follower consensus is achieved. Besides, we note that the slow speed of convergence is due to the sparsity of the graph at each time.
We further conducted another simulation when the reference $r[k]$ varies over time to show the tracking ability of our algorithm. The result is presented in Fig.~\ref{15nodevalue_longtime}. Observe that normal followers can follow the time-varying reference as long as it remains invariant for a sufficiently long period of time.
We also point out that the reference here takes values outside the range of adversarial values and the range of normal nodes' initial values. In contrast, normal nodes using leaderless resilient consensus protocols must reach consensus on a value within the range of their initial values.

\subsection{Simulation in A Time-Varying Second-Order MAS}

Next, we present the simulation of the MDP-MSR algorithm in the time-varying leader-follower network in Fig.~\ref{9node2} under the $1$-local model. We note that our methodology can be extended to decoupled multi-dimensional dynamics of agents. Therefore, in this simulation, each node $i\in \mathcal{V}$ is associated with two dimensions, i.e., $x_i[k]$ and $y_i[k]$. To be specific, each node $i \in \mathcal{V}$ exchanges $\hat{x}_i[k]$ and $\hat{y}_i[k]$ with neighbors, and we employ the MDP-MSR algorithm on each node $i\in \mathcal{W}^\mathcal{N}$ for each axis separately.
Let agents take initial values as $x_i[0]\in (2,5), y_i[0]\in (2,5), \forall i \in \mathcal{W}^\mathcal{N}$, and $r[k]=2, \forall d \in \mathcal{L}^\mathcal{N}, \forall k\geq 0$ is the reference value for both axes. Let $T=0.8$ and $\beta=1.65$, which meet the condition in \eqref{BT}.

The objective is that normal followers track the reference value to form a desired formation such that $\lim_{k\to \infty} \hat{x}_i[k] = \lim_{k\to \infty} \hat{y}_i[k] = r[k], \forall i \in \mathcal{W}^\mathcal{N}$, despite any possible misleading information transmitted by the Byzantine node.
As we discussed in Example~\ref{discussion9node}, this graph is not a jointly $2$-robust following graph with $1$ hop, and hence, is not robust enough to tolerate one Byzantine node using the one-hop MSR algorithms \cite{leblanc2013resilient,ishii2022overview,usevitch2020resilient}. 
Here, we let node 5 be Byzantine. It sends two different oscillating $\hat{x}$ values to nodes in sets $\{1,3,6\}$ and $\{4\}$. Besides, it manipulates its $\hat{y}$ values in the same way.

We first apply the one-hop MDP-MSR algorithm to the time-varying network and observe in Fig.~\ref{9nodevalue2}(a) that resilient leader-follower consensus on the $x$ axis (i.e., $\hat{x}_i[k], i\in \mathcal{W}^\mathcal{N}$) is not achieved. For the one-hop MDP-MSR algorithm to succeed, the network needs more connections. For example, if we add four edges in the network as shown in Fig.~\ref{9node1}, then it becomes a jointly $2$-robust following graph with $1$ hop. We then apply the one-hop MDP-MSR algorithm to the network in Fig.~\ref{9node1} and the results in Fig.~\ref{9nodevalue1} show that resilient leader-follower consensus is achieved. This example also verifies the results in Lemma~\ref{tighter1} and Remark~\ref{tighter35}, which separately state that our graph condition for the one-hop case is tighter than the ones in \cite{usevitch2020resilient,rezaee2021resiliency}. Indeed, the union of the graphs in Fig.~\ref{9node1} does not meet the sufficient conditions in \cite{usevitch2020resilient,rezaee2021resiliency} for the case of static networks.

\begin{figure}[t]
	\centering
	%\vspace{-10pt}
	\subfigure[\scriptsize{The one-hop MDP-MSR algorithm.}]{
		\includegraphics[width=3.3in,height=1.9in]{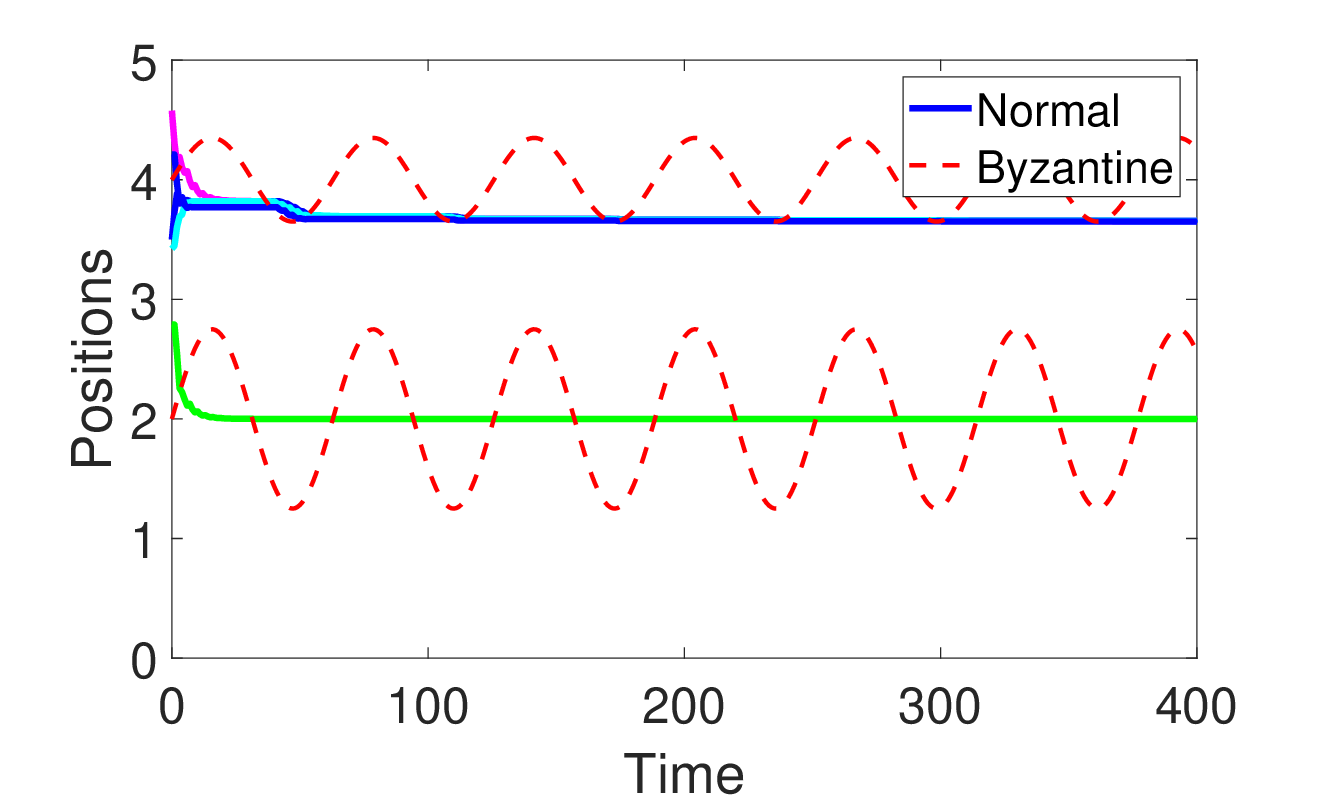}
	}
	\vspace{-3pt}
	
	\subfigure[\scriptsize{The two-hop MDP-MSR algorithm.}]{
		\includegraphics[width=3.3in,height=1.9in]{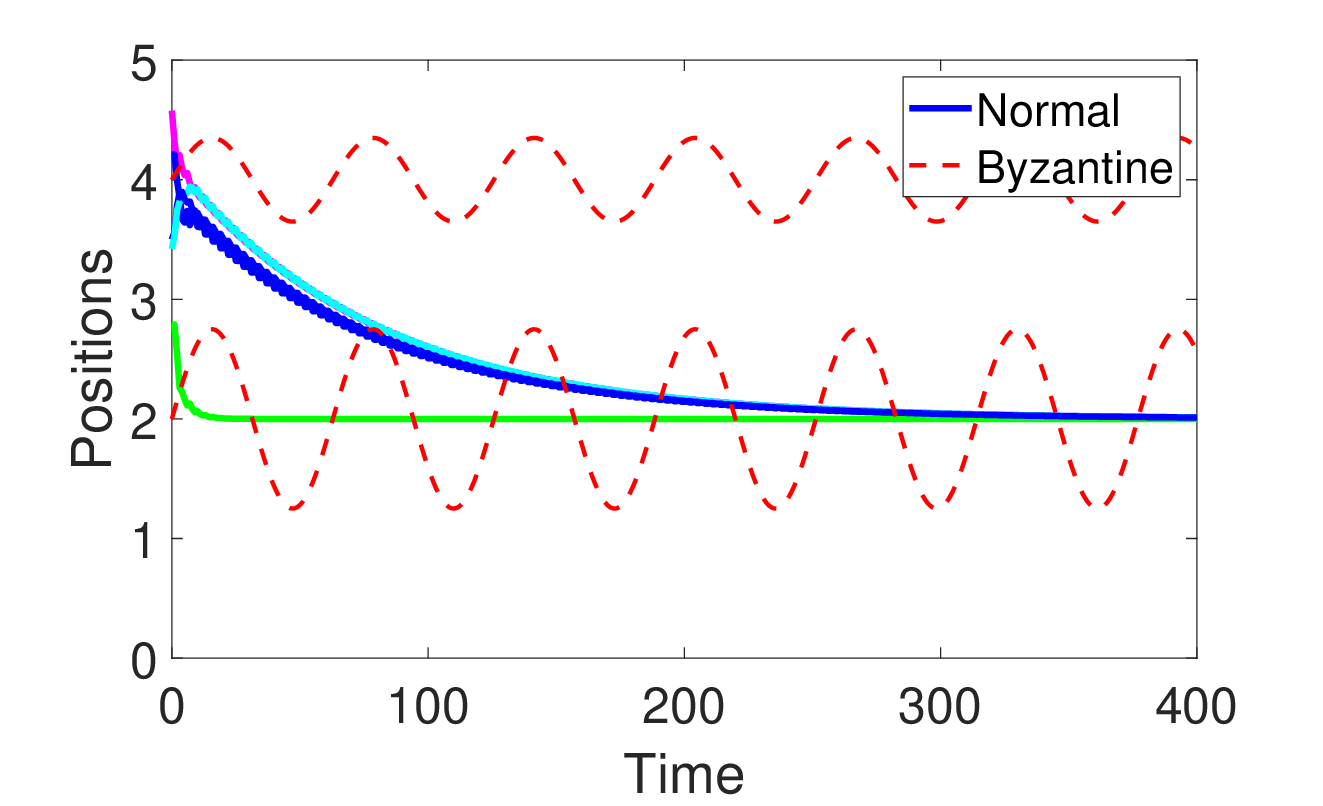}
	}
	\vspace{-3pt}
	\caption{Nodes' positions in the time-varying leader-follower network in Fig.~\ref{9node2} applying the MDP-MSR algorithm.}
	\label{9nodevalue2}
	\vspace*{-2.5mm}
\end{figure}

Alternatively, we can increase the network robustness against adversaries by introducing multi-hop relays without changing the original topology of the network.
This time, we apply the two-hop MDP-MSR algorithm to the network in Fig.~\ref{9node2}. Note that this network becomes a jointly $2$-robust following graph with $2$ hops.
As Theorem~\ref{theorem_secondorder} indicated, with one Byzantine node, it can achieve resilient leader-follower consensus. We should note that the Byzantine node may have more options for attacks under two-hop communication. Thus, we assume that node 5 not only manipulates its own values as before but also relays false information in the same way as it does to its own values. The results of the two-hop algorithm are given in Fig.~\ref{9nodevalue2}(b). Observe that resilient leader-follower consensus on the $x$ axis of normal agents is achieved despite node 5 sending two faulty values to neighbors at each time.

\begin{figure}[t]
	\centering
	%\vspace{-10pt}
	\includegraphics[width=3.3in,height=1.9in]{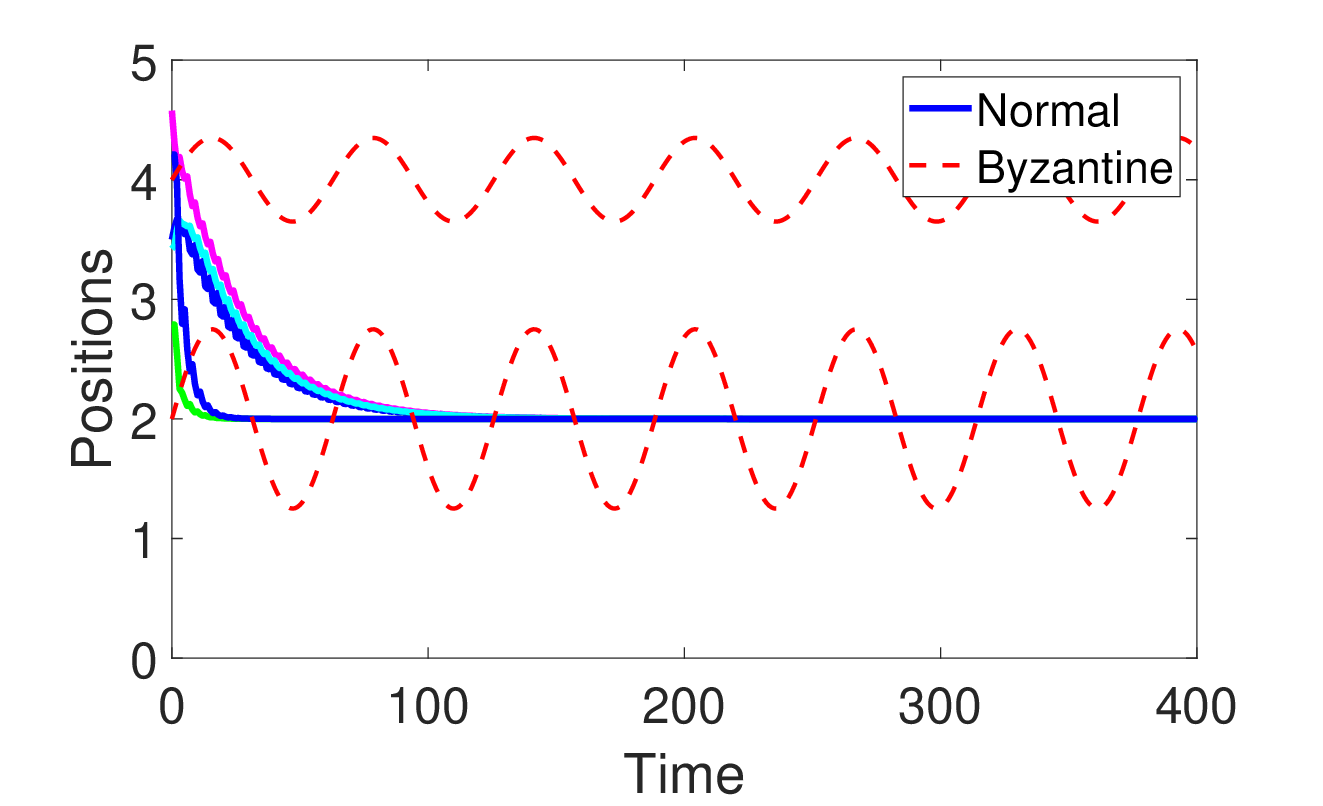}
	\vspace{-3pt}
	\caption{Nodes' positions in the time-varying leader-follower network in Fig.~\ref{9node1} applying the one-hop MDP-MSR algorithm.}
	\label{9nodevalue1}
	\vspace*{-2.5mm}
\end{figure}

Lastly, we present the results for formation control of our algorithm in the network in Fig.~\ref{9node2}. The goal for normal nodes is to form a triangle with three leaders located at one vertex of such a triangle.
First, the trajectories of normal followers applying the one-hop MDP-MSR algorithm are presented in Fig.~\ref{9nodeposition_1hop}. Here, there are two red trajectories of node 5 since it sends two different positions to neighbors. Besides, these positions are oscillating along the two red trajectories in Fig.~\ref{9nodeposition_1hop}. Observe in Fig.~\ref{9nodevalue2}(a) and Fig.~\ref{9nodeposition_1hop} that resilient leader-follower consensus for formation control is not achieved in both dimensions. Next, the two-hop MDP-MSR algorithm is applied under the same scenario, and the nodes' trajectories are presented in Fig.~\ref{9nodeposition_2hop}. It is clear in Fig.~\ref{9nodevalue2}(b) and Fig.~\ref{9nodeposition_2hop} that end positions of normal nodes form the desired formation. Thus, normal nodes using the two-hop algorithm achieve resilient leader-follower consensus on both dimensions despite misbehaviors of node 5.
Through these examples, we have verified the efficacy of the MDP-MSR algorithm.

\section{Conclusion}
In this paper, we have investigated the problem of resilient leader-follower consensus in time-varying networks when multi-hop communication is available.
Our approach is based on the MW-MSR algorithm from our previous work \cite{yuan2021resilient} studying leaderless resilient consensus.
Moreover, we have also studied agents possessing second-order dynamics and proposed the MDP-MSR algorithm to handle resilient leader-follower consensus in time-varying networks.
Our main results have characterized tight necessary and sufficient graph conditions for the proposed algorithms to succeed, which are expressed in terms of jointly robust following graphs with $l$ hops. Our graph condition for static networks is tighter than the one in \cite{usevitch2020resilient} with insecure leaders and the one in \cite{rezaee2021resiliency} with the secure leader when $l= 1$.
With multi-hop communication, we are able to enhance robustness of leader-follower networks without increasing physical communication links and obtain further relaxed graph requirements for our algorithms to succeed. The efficacy of our algorithms has also been verified through extensive numerical examples.
In future work, we intend to tackle resilient dynamic leader-follower consensus using the proposed multi-hop relay techniques, where the leaders may have time-varying reference values.

\begin{figure}[t]
	\centering
	%\vspace{-10pt}
	\includegraphics[width=3.1in]{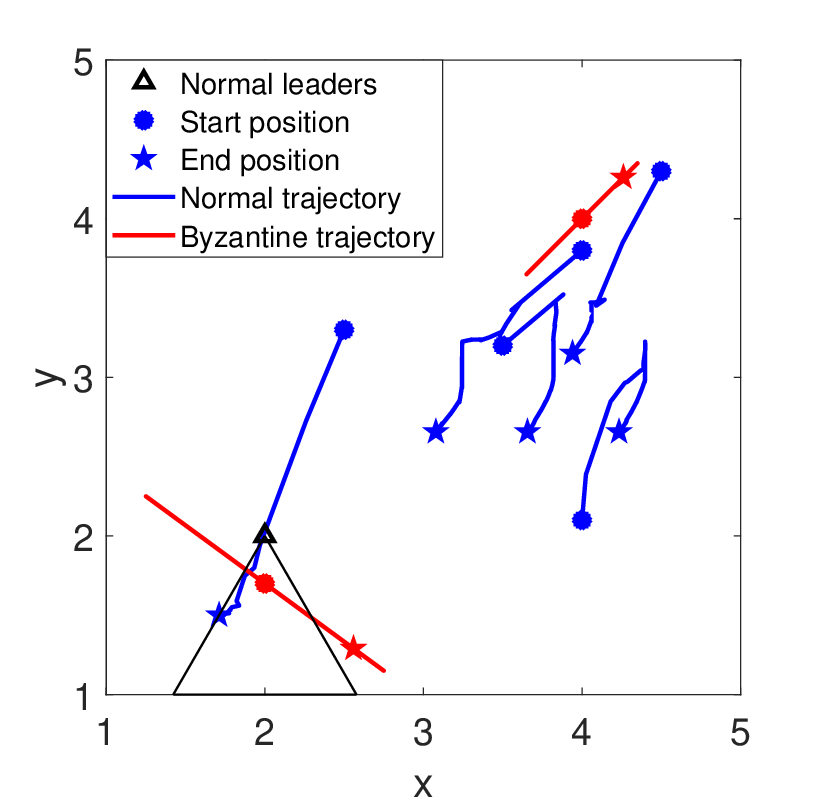}
	\vspace{-3pt}
	\caption{Nodes' trajectories in the time-varying leader-follower network in Fig.~\ref{9node2} applying the one-hop MDP-MSR algorithm.}
	\label{9nodeposition_1hop}
	\vspace*{-2.5mm}
\end{figure}

\appendices
\section*{Appendix}
\section*{A. Proof of Lemma \ref{tighter1}}

\begin{proof}
	If graph $\mathcal{G}[k]$ is strongly $(Q, t_0, 2f+1)$-robust w.r.t. the set $\mathcal{L}$ where $Q=1$,  take a set $\mathcal{F}$ satisfying the $f$-local model. Select any nonempty subset $\mathcal{W}_1\subset \mathcal{H}$, where $\mathcal{H}=\mathcal{W}\setminus\mathcal{F}$. Choose node $i\in \mathcal{Z}_{\mathcal{W}_1}^{2f+1}$. Then, after removing nodes in the set $\mathcal{F}$ from $\mathcal{V}$, at most $f$ follower nodes are removed. Thus,
	it must hold that $i\in \mathcal{Z}_{\mathcal{W}_1}^{f+1}$ in $\mathcal{G}_{\mathcal{H}}$. Hence, $\mathcal{G}_{\mathcal{H}}$ satisfies the condition in Definition~\ref{robust_following}.
	Since this is true for any set $\mathcal{F}$, then $\mathcal{G}[k]$ is a jointly $(f+1)$-robust following graph with $1$ hop where $k_{t+1}-k_t=1,  \forall [k_t, k_{t+1})$.
	
	Yet, the converse statement does not hold. This claim can be proved by simply noticing a counter example in Fig.~\ref{9node1}. There, the union of the three graphs $\mathcal{U}_{\mathcal{G}[k]}$ satisfies our condition in Definition~\ref{robust_following}, and hence, the union graph $\mathcal{U}_{\mathcal{G}[k]}$ is a jointly $(f+1)$-robust following graph with $1$ hop where the length of each interval $[k_t, k_{t+1})$ is $1$.
	Apparently, $\mathcal{U}_{\mathcal{G}[k]}$ is not strongly $(Q, t_0, 2f+1)$-robust w.r.t. the set $\mathcal{L}$ where $Q=1$. 
	We can observe from the fact that in follower set $\mathcal{S}=\{1,2,3\}$, none of nodes has 3 in-neighbors outside $\mathcal{S}$. 
\end{proof}

\begin{figure}[t]
	\centering
	%\vspace{-10pt}
	\includegraphics[width=3.1in]{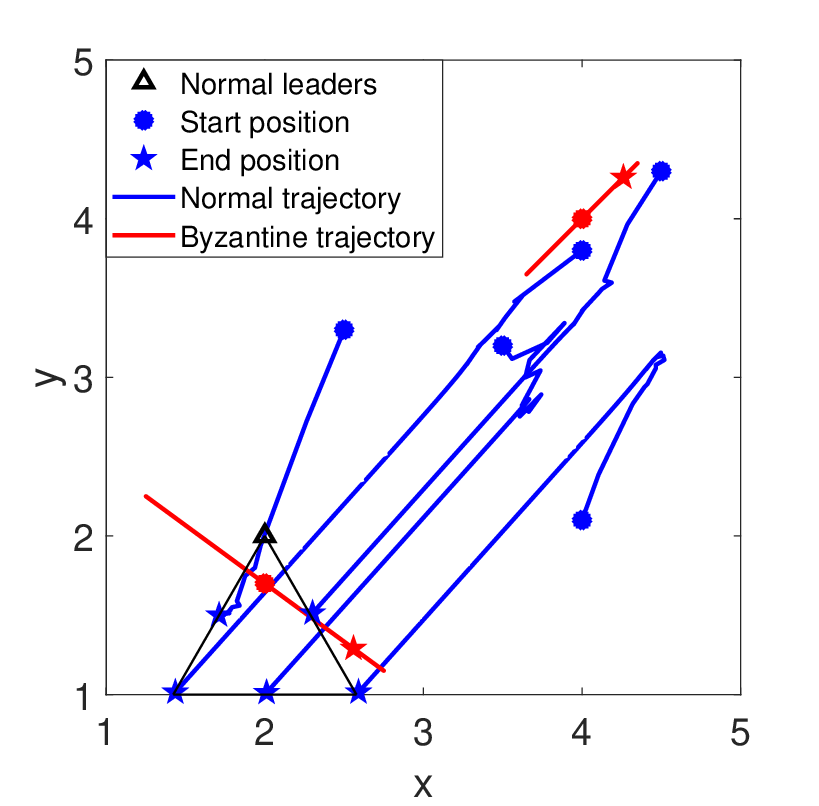}
	\vspace{-3pt}
	\caption{Nodes' trajectories in the time-varying leader-follower network in Fig.~\ref{9node2} applying the two-hop MDP-MSR algorithm.}
	\label{9nodeposition_2hop}
	\vspace*{-2.5mm}
\end{figure}

\section*{B. Proof of Lemma \ref{2f+1leaders}}

\begin{proof}
	1) Since $\mathcal{G}[k]$ is a jointly $(f+1)$-robust following graph with $l$ hops, in the normal network $\mathcal{G}_{\mathcal{N}}[k] = (\mathcal{N},\mathcal{E}_{\mathcal{N}}[k] )$, for node set $\mathcal{S}=\mathcal{W}^\mathcal{N} = \mathcal{N}\setminus \mathcal{L}$, it holds that $| \mathcal{Z}_{\mathcal{S}}^{f+1}[k_t, k_{t+1}) |  =| \{i\in \mathcal{S} :  \exists K_i \in [k_t, k_{t+1}) \medspace \textup{s.t.} \medspace  |\mathcal{I}_{i, \mathcal{S}}[K_i]| \geq f+1 \} | \geq 1$. Moreover, there could be at most $f$ adversarial leaders, and thus, there must be at least $2f+1$ leaders, i.e., $|\mathcal{L}| \geq 2f+1$.
	
	2) This claim is proved by contradiction. If $\exists [k_t, k_{t+1})$ s.t. $\forall k \in [k_t, k_{t+1})$, $\forall i \in \mathcal{W}^\mathcal{N}$,  $|\mathcal{N}_i^{l-}[k] \cap \mathcal{L}| \leq 2f$, and if there are $f$ adversarial leaders, then in the normal network $\mathcal{G}_{\mathcal{N}}[k] = (\mathcal{N},\mathcal{E}_{\mathcal{N}}[k] )$, for node set $\mathcal{W}^\mathcal{N} = \mathcal{N}\setminus \mathcal{L}$, it holds that $\forall k \in [k_t, k_{t+1})$,  $|\mathcal{I}_{i, \mathcal{W}^\mathcal{N}}[k]| \leq f$. Thus, $\mathcal{G}[k]$ cannot be a jointly $(f+1)$-robust following graph with $l$ hops.

	3) It is proved by contradiction. If $|\mathcal{W}_\mathcal{L}[k_t, k_{t+1})| \leq 2f$, and if $f$ nodes in $\mathcal{W}_\mathcal{L}[k_t, k_{t+1})$ are adversarial, then for node set $\mathcal{W}^\mathcal{N} \setminus \mathcal{W}_\mathcal{L}[k_t, k_{t+1})$, it holds that $\forall k \in [k_t, k_{t+1})$, $|\mathcal{I}_{i, \mathcal{W}^\mathcal{N} \setminus \mathcal{W}_\mathcal{L}[k_t, k_{t+1}) } [k]| \leq f$. This contradicts that $\mathcal{G}[k]$ is a jointly $(f+1)$-robust following graph with $l$ hops.

	4) We prove this claim by contradiction. If $\exists [k_t, k_{t+1})$ s.t. $\forall k \in [k_t, k_{t+1})$, $|\mathcal{N}_i^{1-}[k]| \leq 2f$, and if $f$ of such in-neighbors are adversarial, then node $i$ only has $f$ normal in-neighbors. Hence, the set $\{i\}$ does not satisfy the condition in Definition~\ref{robust_following}.
	Thus, $\mathcal{G}[k]$ cannot be a jointly $(f+1)$-robust following graph with $l$ hops. Furthermore, it is now easy to derive the second statement.
\end{proof}

\begin{IEEEbiography}[{\includegraphics[width=1in,height=1.25in,clip,keepaspectratio]{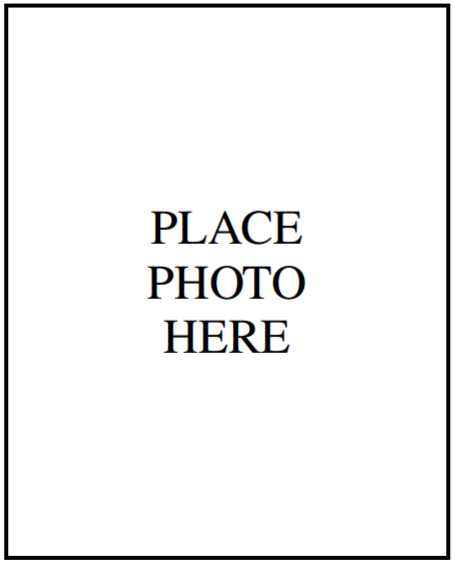}}]{Liwei Yuan} (Member) received the B.E. degree in Electrical Engineering
	and Automation from Tsinghua University,
	China, in 2017, and the Ph.D. degree in Computer
	Science from Tokyo Institute of Technology, Japan, in
	2022.
	He is currently a Postdoctoral Researcher in the College 
	of Electrical and Information Engineering at Hunan 
	University, Changsha, China. His current
	research focuses on security in multi-agent systems
	and distributed algorithms.
\end{IEEEbiography}

\begin{IEEEbiography}[{\includegraphics[width=1in,height=1.25in,clip,keepaspectratio]{blank}}]{Hideaki Ishii} (M'02-SM'12-F'21) received the
	M.Eng.\ degree in applied systems science from
	Kyoto University, Kyoto, Japan, in 1998, and the
	Ph.D. degree in electrical and computer engineering
	from the University of Toronto, Toronto,
	ON, Canada, in 2002. He was a Postdoctoral Research
	Associate at the University of Illinois at Urbana-Champaign,
	Urbana, IL, USA, from 2001 to
	2004, and a Research Associate at
	The University of Tokyo, Tokyo, Japan, from 2004 to 2007.
	He was an Associate Professor and Professor at the Department of Computer Science,
	Tokyo Institute of Technology, Yokohama, Japan in 2007--2024. Currently, he is a Professor at the 
	Department of Information Physics and Computing at The University of Tokyo, Tokyo, Japan.
	He was a Humboldt Research Fellow at the University of Stuttgart
	in 2014--2015. He has also held visiting positions at CNR-IEIIT at
	the Politecnico di Torino, the Technical University of Berlin, and
	the City University of Hong Kong. His research interests
	include networked control systems, multiagent systems, distributed algorithms,
	and cyber-security of control systems.
	
	Dr.~Ishii has served as an Associate Editor for Automatica, 
	the IEEE Control Systems Letters, the IEEE Transactions on Automatic Control, 
	the IEEE Transactions on Control of Network Systems,
	and the Mathematics of Control, Signals, and Systems.
	He was a Vice President for the IEEE Control Systems Society (CSS) in 2022--2023.
	He was the Chair of the IFAC Coordinating Committee on Systems and
	Signals in 2017--2023.
	He served as the IPC Chair for the IFAC World Congress 2023 held in Yokohama, Japan.
	He received the IEEE Control Systems Magazine Outstanding Paper
	Award in 2015. Dr.~Ishii is an IEEE Fellow.
\end{IEEEbiography}

\end{document}